\documentclass[prx,twocolumn,10pt,notitlepage,superscriptaddress,nofootinbib]{revtex4-1}

\usepackage[english]{babel}

\setcitestyle{sort&compress,numbers}
\usepackage{amsthm, color, mathtools, amsmath, amssymb, setspace,bbm, tensor, subfigure, mathtools, placeins, graphicx, wrapfig,float, verbatim, xcolor,enumitem}
\usepackage{graphicx}
\usepackage[unicode=true, breaklinks=false, pdfborder={0 0 1}, backref=false, colorlinks=true, linkcolor=blue, citecolor=blue]{hyperref}
\usepackage{cancel,bm}
\usepackage{bbm}
\usepackage{orcidlink}

\let\oldsqrt\sqrt
\def\sqrt{\mathpalette\DHLhksqrt}
\def\DHLhksqrt#1#2{%
\setbox0=\hbox{$#1\oldsqrt{#2\,}$}\dimen0=\ht0
\advance\dimen0-0.2\ht0
\setbox2=\hbox{\vrule height\ht0 depth -\dimen0}%
{\box0\lower0.4pt\box2}}

\newcommand{\tr}[1]{\operatorname{tr}\left[#1\right]}

\newcommand{\tra}{\operatorname{tr}}

\newcommand{\id}{\mathbbm{1}}

\newcommand{\Jcal}{\mathcal{J}}
\newcommand{\Kcal}{\mathcal{K}}

\newcommand{\Pprob}{\mathbb{P}}

\newcommand{\ident}{\mathbbm{1}}

\def\bra#1{\mathinner{\langle{#1}|}}

\def\ket#1{\mathinner{|{#1}\rangle}}

\def\braket#1{\mathinner{\langle{#1}\rangle}}

\newcommand*\xbar[1]{%
   \hbox{%
     \vbox{%
       \hrule height 0.5pt 
       \kern0.5ex
       \hbox{%
         \kern-0.2em
         \ensuremath{#1}%
         \kern-0.0em
       }%
     }%
   }%
} 

\def\BraVert{\egroup\,\mid\,\bgroup}

\def\ketbra#1#2{\ket{#1\vphantom{#2}}\!\bra{#2\vphantom{#1}}}

\def\bra#1{\mathinner{\langle{#1}|}}

\def\ket#1{\mathinner{|{#1}\rangle}}

\def\braket#1{\mathinner{\langle{#1}\rangle}}

\usepackage{bbm, braket}
\usepackage[mathscr]{euscript}
\allowdisplaybreaks
\makeatletter
\theoremstyle{plain}
\newtheorem{thm}{\protect\theoremname}
\theoremstyle{plain}

\theoremstyle{plain}

\theoremstyle{remark}
\newtheorem*{rem*}{\protect\remarkname}
\theoremstyle{definition}
\newtheorem*{rems*}{Remark}
\theoremstyle{plain}

\theoremstyle{plain}

\theoremstyle{definition}

\theoremstyle{plain}
\newtheorem*{thm*}{\protect\theoremname}
\theoremstyle{plain}
\newtheorem*{lem*}{\protect\lemmaname}
\theoremstyle{plain}

\theoremstyle{definition}
\newtheorem{example}{\protect\examplename}
 
\providecommand{\propositionname}{Proposition}
\providecommand{\theoremname}{Theorem}
\providecommand{\lemmaname}{Lemma}
\providecommand{\remarkname}{Remark}
\providecommand{\conjecturename}{Conjecture}
\providecommand{\definitionname}{Definition}
\providecommand{\corollaryname}{Corollary}
\providecommand{\observationname}{Observation}
\providecommand{\examplename}{Example}
\allowdisplaybreaks

\newcommand{\cv}[1]{(#1)}
\newcommand{\cvb}[1]{[#1]}

\newcommand{\cvV}[1]{\Vert #1\Vert}
\newcommand{\cvr}[1]{\left\langle #1\right\rangle}
\newcommand{\iden}{\id}

\newcommand{\iu}{{i\mkern1mu}}

\begin{document}

\renewcommand{\abstractname}{}

\title{Connecting Commutativity and Classicality for Multi-Time Quantum Processes} 

\author{Fattah Sakuldee\,\orcidlink{0000-0001-8756-7904}}
\email{fattah.sakuldee@ug.edu.pl}
\affiliation{The International Centre for Theory of Quantum Technologies, University of Gda\'nsk, Ba\.zy\'nskiego 1A, 80-309 Gda\'nsk, Poland}

\author{Philip Taranto\,\orcidlink{0000-0002-4247-3901}}
\email{philip.taranto@tuwien.ac.at} 
\affiliation{Atominstitut, Technische Universit{\"a}t Wien, 1020 Vienna, Austria}
\affiliation{University of Vienna, Vienna Doctoral School in Physics, Boltzmanngasse 5, 1090 Vienna, Austria}
\affiliation{Institute for Quantum Optics and Quantum Information, Austrian Academy of Sciences, Boltzmanngasse 3, 1090 Vienna, Austria}

\author{Simon Milz\,\orcidlink{0000-0002-6987-5513}}
\affiliation{Institute for Quantum Optics and Quantum Information, Austrian Academy of Sciences, Boltzmanngasse 3, 1090 Vienna, Austria}

\begin{abstract}
Understanding the demarcation line between classical and quantum is an important issue in modern physics. The development of such an understanding requires a clear picture of the various concurrent notions of `classicality' in quantum theory presently in use. Here, we focus on the relationship between \emph{Kolmogorov consistency} of measurement statistics---the foundational footing of classical stochastic processes in standard probability theory---and the \emph{commutativity} (or absence thereof) of measurement operators---a concept at the core of quantum theory. Kolmogorov consistency implies that the statistics of sequential measurements on a (possibly quantum) system could be explained entirely by means of a classical stochastic process, thereby providing an operational notion of classicality. On the other hand, commutativity of measurement operators is a structural property that holds in classical physics and its breakdown is the origin of the uncertainty principle, a fundamentally quantum phenomenon. Here, we formalise the connection between these two a priori independent notions of classicality, demonstrate that they are distinct in general and detail their implications for memoryless multi-time quantum processes. 
\end{abstract}

\date{\today}
		
\maketitle

\section{Introduction}
Since the inception of quantum theory, various notions of `classicality' for the states of physical systems and measurements thereof have been put forth, including those based on the commutativity of observables~\mbox{\cite{Luders2006, Alicki2008a, Alicki2008b, Facchi2012, Facchi2013, fazio_witnessing_2013,Sakuldee2021}} (the breakdown of which being the origin of Heisenberg's uncertainty principle~\cite{Heisenberg1927,Ozawa2003,Busch2013}), the absence of coherence~\cite{Aberg2006,Baumgratz2014,Levi2014,Winter2016,Streltsov2017Col,Hu2018} or quantum discord~\cite{Zurek2000,Ollivier2001,Modi_2012_Discord}, the non-negativity of the Wigner function~\cite{Dodonov2002,Kenfack2004,Cormick2006,Bohmann_2020_PRL,Bohmann_2020_Quantum}, the broadcastability of states~\cite{Barnum1996,Piani2008}, or the objectivism that emerges through Darwinist arguments~\cite{Zurek2009,Horodecki2015,Le2018,Le2019,Korbicz_2021}. Most of these concepts of classicality are \textit{static} in the sense that they focus on properties of quantum states and/or the compatibility of measurements in situations where there is no dynamics taking place between them. When extending such considerations to \emph{processes}, i.e., physical systems that display non-trivial evolution and are measured at several points in time, classicality is often linked to the inability of a process to generate and/or detect states displaying such aforementioned properties, as well as certain properties of the resulting multi-time statistics (see, e.g., Refs.~\cite{Smirne2019b, Strasberg2019, Milz2019,Smirne2019}). Additionally, for the case of sequential measurements, non-commutativity is the key ingredient in the generalisation of stochastic processes to the quantum realm~\cite{Accardi1982}. However, besides partial results, the links between such---a priori inequivalent---notions of classicality (or non-classicality) remain poorly understood and subject to debate~\cite{Wilde2010,Briggs2011,Miller2012,Leon-Montiel2013,Oreilly2014}, both in static and dynamic scenarios.

Broadly speaking, existing notions of classicality fall into one of two categories: \emph{Structural} ones, i.e., criteria for classicality based on mathematical properties like the commutativity (of observables) or the coherence of quantum states; and \emph{operational} ones, i.e., those based only upon experimentally accessible entities, such as the multi-time statistics obtained from probing an evolving quantum state at different points in time. While both types of considerations are well-motivated in their own right, the connection between such---generally inequivalent---structural and operational notions of classicality has not yet been fully established; in the static case, it is only known for special cases~\cite{Luders2006,Busch1998, Arias2002, Heinosaari2010} and in the dynamic scenarios that we will focus on such links are only known when restricting to projective measurements of a fixed observable~\cite{Smirne2019,Smirne2019b, Strasberg2019, Milz2019}. Here, we establish more general connections between structural and operational notions of classicality for the case of a quantum system that undergoes non-trivial dynamics and is probed at multiple points in time with general instruments. Specifically, we analyse the connection between the satisfaction of so-called \emph{Kolmogorov consistency conditions}---an operational notion of classicality---and commutativity of the operators that `naturally' assume the role of observables in multi-time processes (we motivate and identify these operators in Sec.~\ref{sec:results}).  

Satisfaction of the former criterion implies the existence of a classical (i.e., described by standard probability theory) stochastic process that leads to the same statistics as the one observed when the underlying quantum process is measured~\cite{Milz2020}; in other words, although such a process might actually be quantum in nature, one cannot conclude this from the collected statistical data alone. Importantly, checking the Kolmogorov consistency conditions amounts to a clear operational notion of classicality that can be tested without any additional knowledge of the underlying dynamics or physical theory. For the case of quantum theory restricted to sequential measurements in a fixed basis, this criterion has been connected to the ability of the dynamics to generate and detect coherences, both in the Markovian (memoryless)~\cite{Smirne2019,Smirne2019b} as well as non-Markovian setting~\cite{Strasberg2019,Milz2019}. Extensions to more general measurements have remained elusive to date.

On the other hand, (non-)commutativity of observables---the structural property that we consider here---lies at the core of quantum theory.\footnote{This fact notwithstanding, recently, the possibility of non-commutativity in classical physics has been discussed~\cite{morgan_algebraic_2020}.} Intuitively, commutativity of two observables $A$ and $B$ implies that they are jointly measurable; that is, given an arbitrary quantum state $\rho$, the probability of obtaining an outcome pertaining to observable $B$ is independent of whether $A$ was measured before it or not (and vice versa). This connection between measurement \textit{non-invasiveness}---an operational notion of classicality---and commutativity of observables---a structural notion of classicality---was first considered by L\"uders~\cite{Luders2006} for the case of projective measurements and later extended to more general scenarios~\cite{Busch1998, Arias2002, Heinosaari2010} (see also Sec.~\ref{sec:preliminaries}), where it was shown that commutativity and measurement non-invasiveness coincide in many cases. Such results thereby endow the commutativity of relevant operators with operational meaning in terms of a notion of classicality based upon measurement non-invasiveness. Importantly, such a direct connection between these two a priori distinct concepts can only be meaningfully established under the assumption that there are only \textit{two} sequential measurements being considered.

In the multi-time setting with non-trivial dynamics between measurements, it is a priori unclear how L\"uders' results carry over. In other words, how can one link the structural property of commutativity with an operationally-clear notion of classicality? The first step to doing so is to determine what `observables' are the appropriate ones to consider when checking commutativity. In particular, both the underlying dynamics in between measurements and the effects of general measurement instruments must be accounted for in the temporal setting. While this can be done by combining the chosen measurements with the dynamics, it is then not necessarily the commutativity of the bare measurements (i.e., pertaining to the measurement device itself) per se, but rather the \textit{effective} measurements (i.e., those with the dynamics accounted for) that render the observed statistics `classical' or `non-classical' accordingly. 

Here, we identify the operators that determine the non-invasiveness of measurements for the case of multiple sequential measurements with non-trivial intermediate dynamics and analyse the conditions for which the commutativity---or weaker versions thereof---of said operators corresponds to the satisfaction of the Kolmogorov consistency conditions (and vice versa). Our analysis thus connects structural with operational notions of classicality for multi-time dynamics and general measurement settings. For the special case of two sequential measurements without intermediate dynamics, our results coincide with those of L{\"u}ders. However, in general, the situation presents itself considerably more layered, and a `straightforward' extension of Lüders' results is not possible. We show that commutativity (of the relevant operators) is a stronger condition than Kolmogorov consistency in general; while the former implies the latter, the converse is not true. Additionally, we derive the conditions under which Kolmogorov consistency implies the vanishing of pertinent commutators in a restricted---but still multi-time---setting and highlight the ensuing physical implications in order to develop intuition concerning the interplay of these two notions of classicality. Finally, we relate our considerations to the well-known case of projective measurements in a fixed basis---where structural properties that are equivalent to Kolmogorov consistency have been identified~\cite{Smirne2019}---and show that, while said structural considerations follow directly from those we provide for more general measurement scenarios, even in this special case, it is difficult to identify generally applicable commutator relations. All throughout, we provide counter-examples to existing notions of classicality that serve to help build intuition and highlight the necessity of certain assumptions (e.g., on the types of measurements allowed).

Together, our results offer a comprehensive analysis regarding the connection of structural---yet not directly observable in many cases---properties of quantum dynamics and operational---i.e., experimentally accessible---notions of classicality, and underline the complicated interplay between dynamics and measurements that arises in the multi-time scenario. 

This article is organised as follows. We begin by outlining some preliminary concepts, including the considerations of L{\"u}ders~\cite{Luders2006} that motivate the examination of commutativity, and similarly for the Kolmogorov consistency conditions, in Sec.~\ref{sec:preliminaries}. We then explore the link between these two concepts within the setting of multi-time Markovian quantum dynamics throughout Sec.~\ref{sec:results}, where we first derive a commutator expression whose vanishing is sufficient to imply classical statistics, before deriving a necessary condition for classicality to imply vanishing commutators of the pertinent operators. We subsequently connect our work with the special case of dynamics that do not generate and detect coherences, which constitutes perhaps the most physically relevant special case~\cite{Smirne2019,Smirne2019b, Strasberg2019, Milz2019} that our results apply to. Finally, we conclude with a discussion and outlook in Sec.~\ref{sec:conclusion}.


\section{Preliminaries}\label{sec:preliminaries}

We begin by introducing the relevant concepts for both a structural and an operational definition of classicality in multi-time processes. To this end, first, we recall the connection between commutativity of observables and non-invasiveness in the two-time case, in particular the considerations of L{\"u}ders~\cite{Luders2006}.

\subsection{L{\"u}ders' Theorem: Commutativity and Non-invasiveness}

As a preliminary example concerning the connection between structural and operational notions of classicality, we consider the simplest case: The sequential measurement of two observables $A$ and $B$ on a state $\rho$ without intermediate evolution. Let $\{\Pi^{(a)}\}$ and $\{\Omega^{(b)}\}$ be projectors onto the eigenspaces of $A$ and $B$, respectively, with eigenvalues $\{a\}$ and $\{b\}$. The probability to first measure outcome $a$ and then $b$ is given by
\begin{gather}
    \Pprob(b,a) = \text{tr}(\Omega^{(b)} \Pi^{(a)} \rho \Pi^{(a)}) = \text{tr}(\rho \Pi^{(a)}\Omega^{(b)}\Pi^{(a)})\, .
\end{gather}
In classical physics, future statistics are unaffected by whether or not a previous measurement was conducted\footnote{Importantly, this independence of measurement statistics from previous outcomes generally only holds \textit{on average}, i.e., it does \textit{not} imply $\Pprob(b|a) = \Pprob(b)$, but merely $\Pprob(b) = \sum_a\Pprob(b,a)$.} (when that previous measurement outcome is not recorded, i.e., averaged over). Here, we employ the term `classical' in a somewhat colloquial sense, as pertaining to the macroscopic world; below, we properly define what we mean exactly by `classical' throughout this article. Consequently, if the above situation were classical, then 
\begin{gather}
\label{eqn::KolmoTriv}
    \sum_{a} \Pprob(b,a) = \Pprob(b,\cancel{a})
\end{gather}
would hold, where $\Pprob(b,\cancel{a})$ denotes the probability to obtain outcome $b$ if the first measurement was \textit{not} performed. Note that, crucially, the above equation compares two \emph{distinct} experiments: One in which the system is consecutively measured at \textit{two} times (yielding outcomes $a$ and $b$, l.h.s.); the other in which no measurement is made at the first time (indicated by the slash and yielding only outcome $b$ at the second time, r.h.s.). In quantum mechanics, Eq.~\eqref{eqn::KolmoTriv} (or generalisations thereof, see Sec.~\ref{sec::KolmCons}) fails to hold in general, since quantum measurements are invasive. As a consequence, not performing a measurement is distinguishable from measuring and averaging over outcomes. 

One example where Eq.~\eqref{eqn::KolmoTriv} \textit{can} be satisfied in the quantum setting, independently of $\rho$, is when $[\Pi^{(a)}, \Omega^{(b)}] = 0$ for all $a,b$ (which is equivalent to $[A,B] = 0$), since then $\sum_a \Pi^{(a)} \Omega^{(b)} \Pi^{(a)} = \Omega^{(b)}$ (where we have used $\Pi^{(a)} \Pi^{(a)} = \Pi^{(a)}$ and $\sum_a \Pi^{(a)} = \ident$), and thus 
\begin{gather}
\begin{split}
   \sum_a\Pprob(b,a) &= \sum_a \text{tr}(\Pi^{(a)} \Omega^{(b)} \Pi^{(a)} \rho) \\
   &= \text{tr}(\Omega^{(b)} \rho) =  \Pprob(b,\cancel{a})\, .
 \end{split}
\end{gather}
If the above is satisfied, then---just like in classical physics---all information is contained in the two-point probability distribution $\Pprob(b,a)$ in the sense that both single-time distributions $\Pprob(a)$ and $\Pprob(b)$ can be obtained from it by marginalisation; the process is thereby fully characterised. Consequently, throughout this article, we call \emph{classical} those experimental situations that satisfy this property, i.e., that yield probabilities which can all be obtained from a single multi-time probability distribution via marginalisation (see Sec.~\ref{sec::KolmCons} for a rigorous discussion). In this case, we say that said probability distributions for different subsets of times are \emph{consistent}. Importantly, this notion of classicality amounts to measurement non-invasiveness: Whether or not a measurement has been performed at a given time has no bearing on the outcome probabilities at different times if the statistics are classical. Thus, commutativity of observables $A$ and $B$ in a two-point measurement scheme implies classicality of the statistics (the converse is also true, but not obvious, see below), establishing a direct link between a structural notion (commutativity of operators) of classicality and an operational definition (measurement non-invasiveness) thereof. 

More generally, instead of performing sharp measurements of an observable, an experimenter could first probe the state $\rho$ with a general instrument, described by a set of Kraus operators $\{K^{(a)}\}$, where each $K^{(a)}$ corresponds to a measurement outcome $a$, and subsequently perform a POVM $\{Q^{(b)}\}$, each element of which corresponds to an outcome $b$. The two-point statistics are then given by 
\begin{align}
\label{eqn::Two_point}
    \Pprob(b,a) &= \text{tr}(Q^{(b)} K^{(a)} \rho K^{(a)\dagger}) = \text{tr}(\rho K^{(a)\dagger} Q^{(b)} K^{(a)}) \notag \\
    &=:  \text{tr}(\rho \Kcal^{(a)\dagger} [Q^{(b)}])\, .
\end{align}
Here, $K^{(a)}$ and $Q^{(b)}$ respectively play analogous roles to $\Pi^{(a)}$ and $\Omega^{(b)}$ in the previous example. Now, setting $\Kcal^\dagger[\bullet] : = \sum_a \Kcal^{(a)\dagger}[\bullet]$, we see that non-invasiveness of the first measurement amounts to the satisfaction of
\begin{gather}\label{eq:luedersnoninvasiveness}
    \text{tr}(\rho \Kcal^\dagger[Q^{(b)}]) = \text{tr}(\rho \, Q^{(b)}) \quad \forall \, b\, .
\end{gather}
If the above must hold for \textit{all} states $\rho$, then non-invasiveness is equivalent to 
\begin{gather}
\Kcal^\dagger[Q^{(b)}] = Q^{(b)} \quad \forall \, b\, .
\end{gather}
This criterion has been connected to commutation relations by L\"uders, yielding the following theorem:
\begin{thm}[L\"uders~\cite{Luders2006, Busch1998}]
\label{thm::Luders}
Let $\Kcal^\dagger$ be defined as above and let $Q$ be a positive semi-definite operator. If all Kraus operators $K^{(a)}$ are Hermitian, then $\mathcal{K}^\dagger[Q] = Q$ is equivalent to $[K^{(a)}, Q] = 0 \ \forall \, a$.
\end{thm}
Since some of our proofs below follow a similar line of reasoning, we recall the proof of this theorem from the literature (see e.g. Ref.~\cite{Busch1998}).
\begin{proof}
First, it is easy to see that $[K^{(a)}, Q] = 0 \ \forall \, a$ implies $\mathcal{K}^\dagger[Q] = Q$. To see the converse, let $\ket{\varphi}$ be an arbitrary vector in the Hilbert space $\mathcal{H}$ that $Q$ is defined on. Suppose that $\mathcal{K}^\dagger\cvb{Q}=Q,$ and decompose $Q=\sum_\mu\lambda_\mu P_\mu$ with $\lambda_1 > \lambda_2 > \ldots$ with $\{P_\mu\}$ being mutually orthogonal projection operators. It follows that
	\begin{align}\label{eq:LuederOriginal}
		\lambda_1\cvV{P_1\ket{\varphi}}^2 &= \braket{P_1\varphi|QP_1\varphi} \notag \\
			&= \braket{P_1\varphi|\mathcal{K}^\dagger\cvb{Q}P_1\varphi}\notag \\
			&= \sum_a \braket{K^{(a)}P_1\varphi|Q|K^{(a)}P_1\varphi}\notag \\
			&\leqslant \lambda_1 \sum_a \braket{K^{(a)}P_1\varphi|K^{(a)}P_1\varphi}\notag \\
			&=  \lambda_1 \sum_a \braket{K^{(a)\dagger} K^{(a)} P_1\varphi|P_1\varphi} \notag \\
			&=\lambda_1\cvV{P_1\ket{\varphi}}^2 \, , 
	\end{align}
where we have used $\lambda_1\mathbbm{1} - Q \geqslant 0$ for the first inequality and $\sum_a K^{(a)\dagger} K^{(a)} = \ident$ for the last equality. From the above, we see that 
	\begin{align}
		&\braket{K^{(a)}P_1\varphi|(\lambda_1\iden-Q)|K^{(a)}P_1\varphi}  \nonumber \\
		&\phantom{asdfasdfasdf}= \cvV{\left(\lambda_1\iden-Q\right)^{1/2}K^{(a)}P_1\ket{\varphi}}^2
		= 0\label{eq:Luder-train_identity}
	\end{align}
holds, and we thus have
	\begin{equation}
		QK^{(a)}P_1\ket{\varphi} = \lambda_1K^{(a)}P_1\ket{\varphi},\label{eq:Luder-train_Km-invariant}
	\end{equation}
implying that $K^{(a)}$ leaves the $\lambda_1-$eigensubspace invariant. This means that $K^{(a)}P_1=P_1K^{(a)}P_1,$ and since $K^{(a)}$ is assumed to be Hermitian, it follows that $\cvb{K^{(a)},P_1}=0$ for all $a.$ Now, we set $Q_\mu=Q_{\mu-1}-\lambda_{\mu-1} P_{\mu-1}$ and $Q_0=Q,$ and repeat the same steps as above with $Q_2$ and $P_2$ and so on. This iteration then leads to the fact that $\cvb{K^{(a)},P_\mu}=0$ for all $a$ and $\mu.$ Hence $\cvb{K^{(a)},Q}=0$ as claimed.
\end{proof}
Since $\mathcal{K}^\dagger[Q^{(b)}] = Q^{(b)}$ is equivalent to non-invasiveness of the first measurement, Thm.~\ref{thm::Luders} says that non-invasiveness for arbitrary initial states $\rho$ is equivalent to commutativity of the Kraus operators of the first measurement and the POVM elements of the second one, if \textit{all} Kraus operators are Hermitian (which is, e.g., the case for projective measurements of two observables).

We emphasise that Hermiticity of the Kraus operators is crucial for the derivation of Thm.~\ref{thm::Luders}, and without this assumption, it no longer holds in general. This can be seen by considering a counterexample provided in Ref.~\cite{Heinosaari2010}: \\ \\
\begin{example} 
\label{eqn::ExNonHerm}
Let the POVM elements $\{Q^{(1)},Q^{(2)}\}$ be given by 
	\begin{align*}
	Q^{(1)} &= \dfrac{1}{2}\left(\begin{array}{ccc}
		2 & 0 & 0 \\ 
		0 & 0 & 0 \\ 
		0 & 0 & 1
		\end{array}  \right)\!,~Q^{(2)} = \dfrac{1}{2}\left(\begin{array}{ccc}
		0 & 0 & 0 \\ 
		0 & 2 & 0 \\ 
		0 & 0 & 1
		\end{array}  \right),
	\end{align*}
and the Kraus operators by
	\begin{align*}
	K^{(1)}\!&=\!\dfrac{1}{2}\!\left(\begin{array}{ccc}
		\sqrt{2} & 0 & -1 \\ 
		0 & 0 & 0 \\ 
		0 & 0 & 0
		\end{array}  \right)\!,K^{(2)}\!=\!\dfrac{1}{10}\!\left(\begin{array}{ccc}
		0 & 0 & 0 \\ 
		0 & -\sqrt{10} & 2\sqrt{10} \\ 
		0 & 0 & 0
		\end{array}  \right)\!,\\
	K^{(3)} &= \dfrac{1}{2}\left(\begin{array}{ccc}
		0 & 0 & 0 \\ 
		0 & \sqrt{2} & 0 \\ 
		0 & 0 & 0
		\end{array}  \right)\!,~K^{(4)} = \dfrac{1}{20}\left(\begin{array}{ccc}
		0 & 0 & 0 \\ 
		0 & 4\sqrt{10} & 2\sqrt{10} \\ 
		0 & 0 & 0
		\end{array}  \right)\!,\\
	K^{(5)} &= \dfrac{1}{2}\left(\begin{array}{ccc}
		\sqrt{2} & 0 & 1 \\ 
		0 & 0 & 0 \\ 
		0 & 0 & 0
		\end{array}  \right).
	\end{align*}
For this setup, it is straightforward to see that both $Q^{(1)}$ and $Q^{(2)}$ are invariant under $\sum_{a=1}^5K^{(a)\dagger}\bullet K^{(a)}$, i.e., the measurement is non-invasive overall and therefore the resulting statistics are classical, but, for example, $ \cvb{K^{(5)},Q^{(1)}} \neq 0$. Consequently, classicality of statistics and commutativity of the Kraus operators are generally inequivalent notions. Nonetheless, the direct connection between commutativity and measurement non-invasiveness in quantum mechanics, i.e., that provided by L{\"u}ders' theorem, has subsequently been extended to more general (e.g., non-Hermitian) Kraus operators in certain circumstances~\cite{Busch1998, Arias2002, Heinosaari2010}.\hfill $\blacksquare$
\end{example} 

Importantly for our purposes, Thm.~\ref{thm::Luders} identifies the relevant operators whose commutation relations are related to non-invasiveness. As a first step, in what follows we will investigate which operators play the roles of $K^{(a)}$ and $Q^{(b)}$ in the multi-time---i.e., more than two consecutive measurements---case. 

Before doing so, we first note that Thm.~\ref{thm::Luders} (and its extensions) are restricted in their realm of application. Firstly, they are limited to only two sequential measurements with no intermediate evolution.\footnote{This is, technically, not a strong restriction, since any intermediate dynamics could be absorbed into the Kraus operators/POVM elements.} Additionally, the (potential) equivalence between $\Kcal^\dagger[Q^{(b)}] = Q^{(b)}$ and classicality requires non-invasiveness for \textit{all} states $\rho$. As we will discuss, in the multi-time scenario, one is not always guaranteed to have access to a full basis of quantum states at each interrogation time of interest. Consequently, in the multi-time case, the relation between compatible statistics and the commutativity of some appropriate operators presents itself as a more layered issue than in the static or two-time cases, even when the Kraus operators of performed measurements are limited in a similar way to the assumptions of Thm.~\ref{thm::Luders}. 

We now move to consider the operational notion of non-invasiveness in the multi-time case, namely the \textit{Kolmogorov consistency conditions}, which are naturally suited to analysing the classicality (or not) of general quantum processes. We will see that L{\"u}ders' considerations amount to a special case of Kolmogorov consistency, before moving on to develop multi-time `L{\"u}ders-type' theorems, in the sense that they connect non-invasiveness of measurements to the vanishing of a set of pertinent commutator expressions.

\FloatBarrier

\subsection{Kolmogorov Consistency and Non-invasiveness}
\label{sec::KolmCons}

In Eq.~\eqref{eqn::KolmoTriv}, we provided an experimentally accessible notion of non-invasiveness---and thus classicality---for two sequential measurements. The natural way to extend this definition to the multi-time case is as follows. Suppose that an agent probes a physical system at $n$ discrete points in time, recording the corresponding joint probability distribution $\mathbbm{P}(m_n, \hdots, m_1)$ over possible outcomes $\{ m_n, \hdots, m_1 \}$ observed at the respective (strictly decreasing) times $\{ t_n, \hdots, t_1 \}$ [see Fig.~\ref{fig::Kolmo1}] 
Note that, as a part of the definition, we consider the sequence of times at which the system is measured as fixed and given, and throughout we use the shorthand subscript notation $m_{i}$ to denote a measurement at time $t_i$ with outcome $m_i$. Importantly, analogous to the two-time case discussed above, for any classical stochastic process, the recorded probability distribution is guaranteed to satisfy the \emph{Kolmogorov consistency conditions}~\cite{Kolmogorov_1956}, illustrated in Figs.~\ref{fig::Kolmo2} and~\ref{fig::Kolmo3}: 
\begin{align}\label{eq:kolmogorovconsistency}
    \mathbbm{P}(m_n, &\hdots, \cancel{m_i}, \hdots, m_1) \notag \\* &= \sum_{m_i} \mathbbm{P}(m_n, \hdots, m_i, \hdots, m_1) \quad \forall \, i\, ,
\end{align}
where, again, the `crossed out' outcome notation indicates that \textit{no} measurement was performed at the corresponding time ($t_i$ in the above equation).
Evidently, Eq.~\eqref{eq:kolmogorovconsistency} is a multi-time generalisation of the two-time scenario considered in Eq.~\eqref{eqn::KolmoTriv}. The distribution on the l.h.s. corresponds to what an experimenter observes if they do \textit{not} perform any interrogation at time $t_i$, whereas on the r.h.s. the full statistics is recorded and then marginalised over at time $t_i$. As in the two-time case, Kolmogorov consistency states that there is no difference between not having performed a measurement and measuring but averaging over all possible outcomes at any time, corresponding to a sensible notion of classicality in terms of measurement non-invasiveness: For classical stochastic processes, measurements simply reveal a pre-existing property of the system, in line with the assumptions of macroscopic realism used, e.g., in the derivation of Leggett-Garg inequalities~\cite{Leggett1985}. This property fails to hold for quantum processes, since quantum measurements generically alter the state of the system being measured~\cite{Piron1981}.

\begin{figure}[t]
\centering
 \subfigure[~]
 {
  \includegraphics[width=0.95\linewidth]{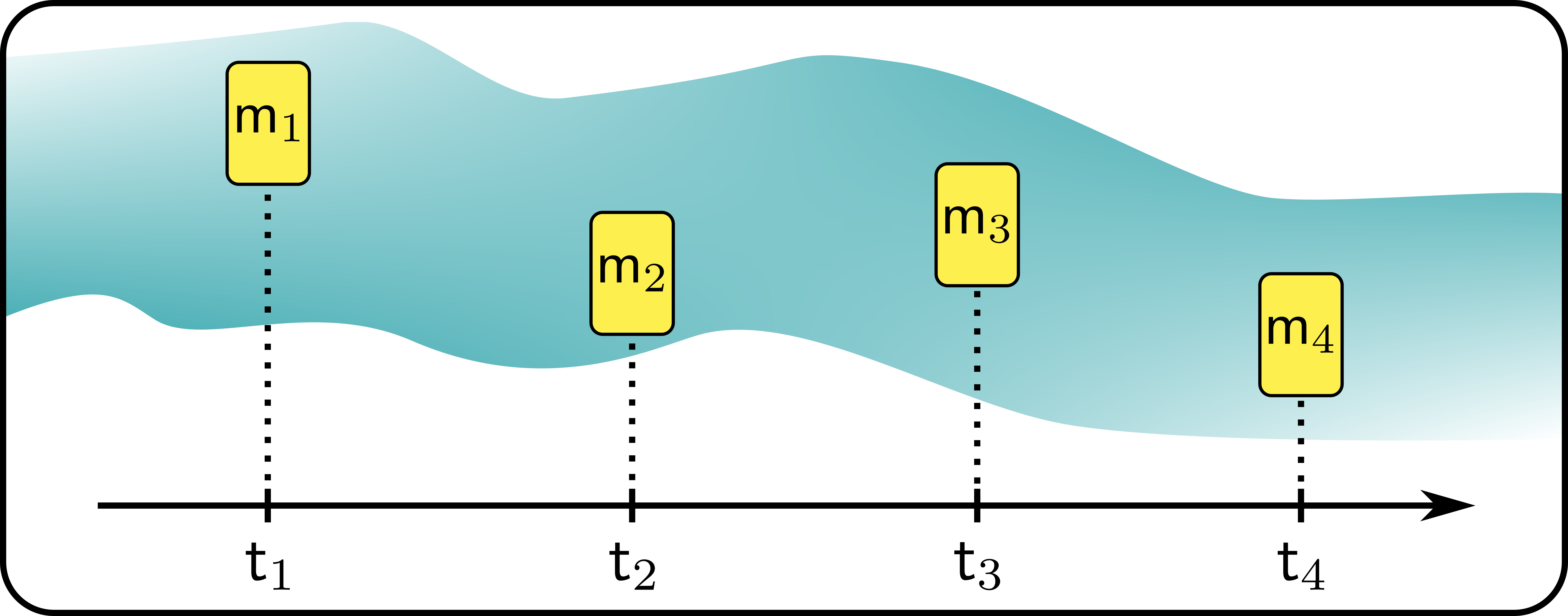}
 \label{fig::Kolmo1}}\\
\subfigure[~] 
{ 
\includegraphics[width=0.95\linewidth]{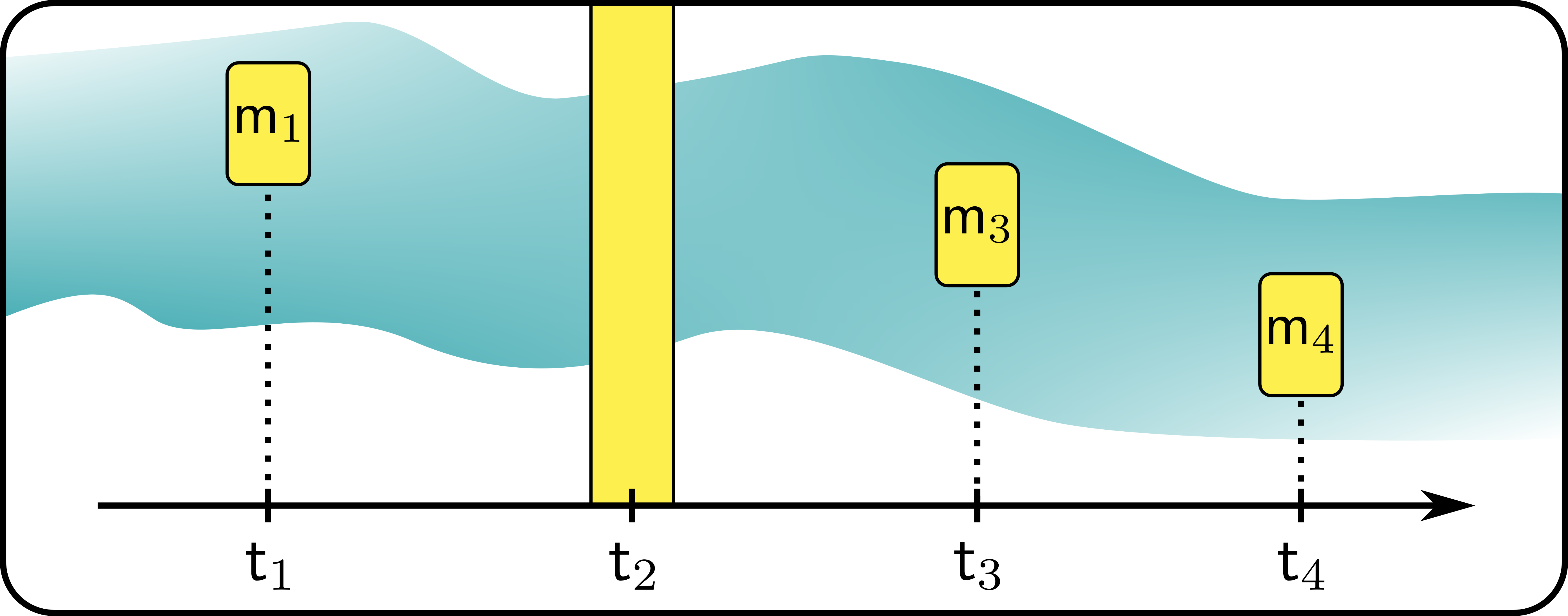}
\label{fig::Kolmo2}}
\\
\subfigure[~] 
{ 
\includegraphics[width=0.95\linewidth]{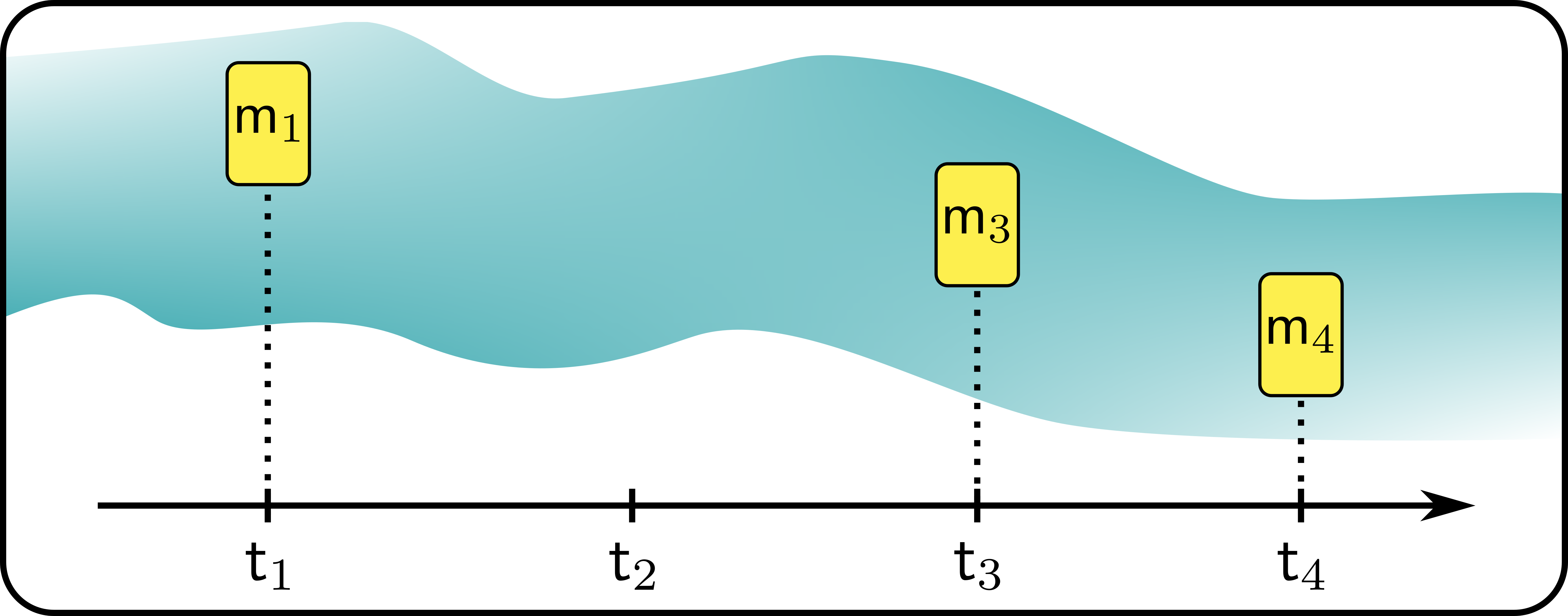}
\label{fig::Kolmo3}}\vspace{-1em}
\caption{\textbf{Schematics of physical processes considered in this article.} An unknown process (turquoise) is probed at different times, here $\{t_1,t_2,t_3,t_4\}$, and the resulting joint probability distribution $\mathbbm{P}(m_4,m_3,m_2,m_1)$ is recorded \textbf{(a)}. From the four-time distribution obtained, one could compute a three-time joint distribution by marginalising over the outcomes at a given time \textbf{(b)}: We depict the case in which the outcomes at time $t_2$ are marginalised over, yielding the probability distribution $\sum_{m_2} \mathbbm{P}(m_4,m_3,m_2,m_1)$ for times $t_1, t_3$, and $t_4$. In contrast, the three-time joint probability $\Pprob(m_4,m_3,\cancel{m_2},m_1),$ as shown in \textbf{(c)}, is obtained by \textit{not} performing a measurement at $t_2$. In general, this is a different experiment than the one of \textbf{(b)} and therefore leads to a different probability distribution. However, if the process is classical---and the employed measurements are non-invasive---then the joint probabilities of the experiments depicted in \textbf{(b)} and \textbf{(c)} coincide [see Eq.~\eqref{eq:kolmogorovconsistency}].}\vspace{-1em}
\end{figure}

We emphasise that a breakdown of Kolmogorov consistency does not necessarily imply that the probed process at hand is non-classical per se. For instance, in the theory of classical causal modelling~\cite{Pearl}, where invasive interrogations can be implemented (e.g., by first measuring the value of some property and then setting it to some other value) in order to potentially infer causal influence, the recorded statistics generally do not satisfy the Kolmogorov consistency conditions~\cite{Milz2020}. Nonetheless, testing the validity of Eq.~\eqref{eq:kolmogorovconsistency} provides a theory-independent, operational procedure to decide on the non-invasiveness of interrogations. In particular, satisfaction of the Kolmogorov consistency conditions implies that there exists a---potentially exotic---classical stochastic process that can reproduce the observed statistics. To do so, said classical stochastic process merely needs to correctly recreate the full joint probability distribution $\Pprob(m_n, \dots, m_1)$, and, due to satisfaction of the Kolmogorov conditions, it then also correctly recreates all joint probability distributions for any subset of times, thereby fully characterising the process from an operational standpoint. Thus, we will interchangeably use the terms `Kolmogorov consistency', `measurement non-invasiveness' and `classicality'.

Recently, in Refs.~\cite{Smirne2019b, Smirne2019, Strasberg2019, Milz2019}, the implications of the satisfaction of the Kolmogorov consistency conditions for general multi-time processes (including those with memory) that are probed by means of pure projective measurements have been characterised, thus connecting the operational, experimentally accessible notion of classicality with certain properties of the underlying quantum dynamics, namely their ability to generate and detect coherence or discord with respect to the chosen measurement basis (we discuss the relationship of these results with our present work in Sec.~\ref{sec::NCGD}). Here, we allow for \emph{general measurements} and phrase our results in terms of commutation relations (i.e., in the spirit of L\"uders' theorem), rather than in terms of the coherence- or discord-related properties of the underlying quantum maps that engender the observed statistics.

\FloatBarrier

\subsection{Multi-time Statistics from (Markovian) Quantum Processes}
\label{subsec:seq-proc}

To make the relation between commutation relations and non-invasiveness in quantum theory more concrete and identify the relevant operators, we now examine how observed statistics are related to the underlying dynamics of a quantum process. In order to collect joint statistics at times $t_1, \dots, t_n$, at each time $t_i$ an experimenter probes the system of interest with an \emph{instrument}, $\mathcal{J}_i = \{\mathcal{K}_i^{(m_i)}\}$, which is a collection of completely-positive \textbf{(CP)} maps that sum up to a completely-positive and trace-preserving \textbf{(CPTP)} map, i.e., $\mathcal{K}_i := \sum_{m_i} \mathcal{K}_i^{(m_i)}$ is a CPTP map~\cite{Lindblad1979}. Each CP map $\mathcal{K}_i^{(m_i)}$ corresponds to a possible outcome $m_i$ and captures the state change of the system upon measurement. For simplicity, we assume that every element $\mathcal{K}_i^{(m_i)}$ can be represented by a single Kraus operator, i.e., $\mathcal{K}_i^{(m_i)}[\rho] = K_i^{(m_i)} \rho K_i^{(m_i)\dagger}$. In between these measurements (e.g., between $t_i$ and $t_{i+1}$), the system of interest undergoes non-trivial dynamics, possibly interacting with an environment, described by CPTP maps $\Lambda_{i+1:i}$. 
While, in principle, both the dynamics of the system, as well as the measurements themselves could be correlated in time, throughout, we focus on \textit{Markovian} (i.e., memoryless) dynamics and uncorrelated measurements. Consequently, the only possible correlations of the observed measurement statistics are then caused by the interplay between the the measurement maps and the dynamical maps (as we will elaborate upon below).\footnote{These assumption may not be fulfilled in the real experiments, and going beyond them leads to interesting phenomena. For instance, conducting two consecutive measurements with the same detector might lead to statistical correlations between the two outcomes beyond what arise from the measured process itself \cite{Lyons1988,valassi2003}; on the other hand, additional memory in the process can, e.g., result in the violation of Leggett-Garg type inequalities in classical processes~\cite{Montina2012}. Our assumptions amount to the---physically reasonable---case of independent measurement devices and a temporal spacing of measurements that is much larger than the typical timescales at which memory effects decay.}

Concretely, assuming the dynamics to be Markovian and the measurements uncorrelated, then the maps $\Lambda_{i+1:i}$ are mutually independent and act on the system alone (see Fig.~\ref{fig::Markov}). Any statistics observed by probing (in an uncorrelated manner) a Markovian process can be computed via the \emph{quantum regression formula}~\cite{Lax1968,Carmichael1993,Breuer_book}: 
\begin{align}\label{eq:qrf}
    \mathbbm{P}&(m_n,\hdots,m_1|\Jcal_n, \hdots, \Jcal_1) = \notag \\
    &\mathrm{tr}\left(\Kcal_n^{(m_n)} \circ \Lambda_{n:n-1} \circ \dots \circ \Lambda_{2:1} \circ \Kcal_1^{(m_1)}[\rho]\right)\, ,
\end{align}
where all maps act on the system alone. It is important to note that the statistics observed depend on \textit{both} the CP maps $\Kcal_i^{(m_i)}$ implemented by the experimenter and the generally uncontrollable dynamics of the process given by the CPTP maps $\{\Lambda_{i:i-1}\}$. Whether or not the measured statistics satisfies the Kolmogorov consistency conditions thus depends on the complex interplay between measurements \textit{and} intermediate dynamics. With Eq.~\eqref{eq:qrf} at hand, we can now identify the relevant operators and commutation relations concerning the satisfaction of Kolmogorov conditions.

\begin{figure}[t]
    \centering
    \includegraphics[width=0.95\linewidth]{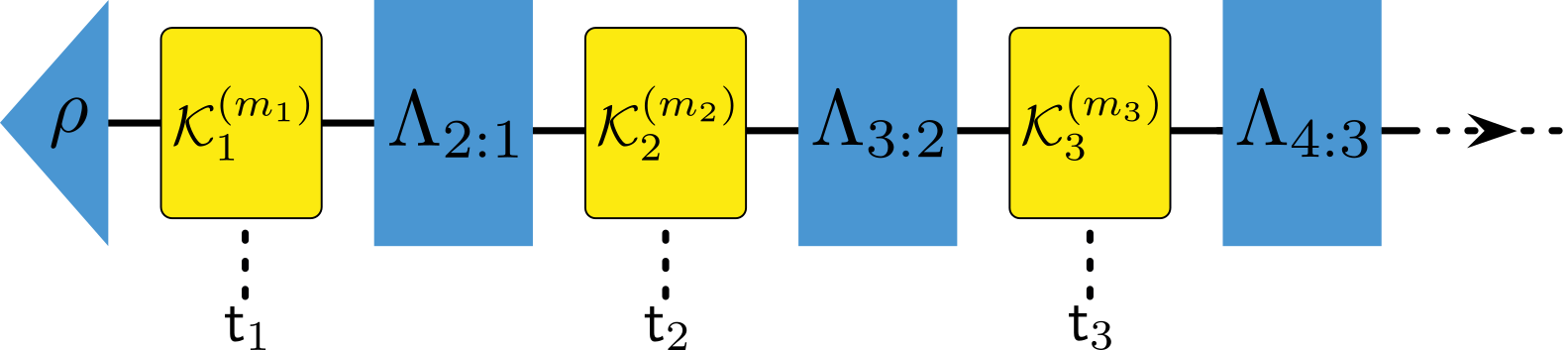}
    \caption{\textbf{Multi-time probing of a Markovian quantum process.} A quantum process without memory on a discrete set of times $t_1, \hdots, t_n$ can be  described by an initial state $\rho$ of the system and a collection of independent CPTP maps $\{\Lambda_{i:i-1}\}$ that act on the system alone between neighbouring times (blue). At each time $t_i$, we envisage an agent probing the process and observing a measurement outcome $m_i$, with the post-measurement state feeding forward (yellow). Such a probing is represented by a CP map $\Kcal_{i}^{(m_i)}$ at each time.}
    \label{fig::Markov}
\end{figure}

\section{Multi-time Dynamics: Kolmogorov Consistency and Commutativity}\label{sec:results}

To identify the relevant commutation relations, let us rewrite Eq.~\eqref{eq:qrf} entirely in terms of Kraus operators: 
\begin{align}
\label{eq:qrf_Kraus}
&\mathbbm{P}(m_n,\hdots,m_1|\Jcal_n, \hdots, \Jcal_1) \notag \\*
&= \sum_{\ell_2\dots \ell_n}\text{tr}\left[ K_n^{(m_n)} L^{(\ell_n)}_{n:n-1}  \cdots L^{(\ell_2)}_{2:1} K_1^{(m_1)} \rho  K_1^{(m_1)\dagger} L^{(\ell_2)\dagger}_{2:1} \cdots\right. \notag \\*
&\phantom{= \sum_{\ell_2\dots \ell_n}}\cdots  \left.L^{(\ell_n)\dagger}_{n:n-1} \Kcal_n^{(m_n)\dagger}\right]\, , 
\end{align}
where we have set $\Lambda_{i:i-1}[\bullet] = \sum_{\ell_i} L_{i:i-1}^{(\ell_i)} \bullet L_{i:i-1}^{(\ell_i)\dagger}$. 

In order to connect this equation to non-invasiveness of a measurement at a time $t_i$, we split the above expression into three parts (using the cyclicity of the trace): One that corresponds to the state immediately prior to the measurement, one corresponding to the measurement itself, and one corresponding to everything that happens after the measurement at said time of interest. Specifically, setting $R_{i:i-1}^{(\ell_i,m_{i-1})}:= L^{(\ell_i)}_{i:i-1} K_{i-1}^{(m_{i-1})}$, we see that the pre-measurement (subnormalised) state at time $t_i$ is given by 
\begin{gather}
\begin{split}
    &\widetilde \rho_i(\mathbf{m}_{i-1:1}) \\
    &:= \sum_{\ell_2\dots \ell_i} R^{(\ell_i,m_{i-1})}_{i:i-1} \cdots R^{(\ell_2,m_1)}_{2:1} \rho  R^{(\ell_2,m_1)\dagger}_{2:1} \cdots R^{(\ell_i,m_{i-1})\dagger}_{i:i-1}\, 
\end{split}
\end{gather}
for $i\geqslant 2$ (with $\widetilde{\rho}_1:=\rho$). Note that $\widetilde \rho_i(\mathbf{m}_{i-1:1})$ depends upon all measurement outcomes $\mathbf{m}_{i-1:1} := (m_{i-1}, \dots, m_1)$ up to $t_i$, and its trace corresponds to the probability to observe said sequence of outcomes, i.e., 
\begin{gather}
    \Pprob(m_{i-1}, \dots, m_{1}| \Jcal_{m-1}, \dots, \Jcal_1) =  \text{tr}[\widetilde \rho_i(\mathbf{m}_{i-1:1})]\,.
\end{gather}
On the other hand, grouping all post-measurement operators together, we can define the positive semi-definite operator
\begin{align}
    &Q_{i}(\mathbf{m}_{n:i+1}) \notag \\*
    &:= \sum_{\ell_{i+1}\dots \ell_n} L_{i+1:i}^{(\ell_{i+1})\dagger} R_{i+2:i+1}^{(\ell_{i+2}, m_{i+1})\dagger} \cdots R_{n:n-1}^{(\ell_n,m_{n-1})\dagger} K_{n}^{(m_n)\dagger} \notag
    \\* &\phantom{:= } \times K_{n}^{(m_n)} R_{n:n-1}^{(\ell_n,m_{n-1})} \cdots  R_{i+2:i+1}^{(\ell_{i+2}, m_{i+1})} L_{i+1:i}^{(\ell_{i+1})}\, .
\end{align}
With this, the multi-time statistics of Eq.~\eqref{eq:qrf_Kraus} can be expressed succinctly as 
\begin{align}
    &\mathbbm{P}(m_n,\hdots,m_1|\Jcal_n, \hdots, \Jcal_1) \notag \\
    &= \text{tr}\left[\widetilde \rho_i(\mathbf{m}_{i-1:1}) K^{(m_i)\dagger}_i Q_{i}(\mathbf{m}_{n:i+1})  K^{(m_i)}_i \right] \notag \\
    &=: \text{tr}\left[\widetilde \rho_i(\mathbf{m}_{i-1:1}) \Kcal^{(m_i)\dagger}_i [Q_{i}(\mathbf{m}_{n:i+1})]\right] \, .
\end{align}
Intuitively, $\widetilde \rho_i$ is the time-evolved (subnormalised) state that is to be measured at time $t_i$, while $Q_{i}$ corresponds to the effect of each measurement outcome after $t_i$ (all the way up until some fixed final time $t_n$), with the dynamics of the process in between accounted for, `rolled back' in time to $t_i$, when the measurement described by $\Kcal^{(m_i)\dagger}_i$ occurs.\footnote{For a slightly different perspective on how one could model sequential measurements, see, e.g., Ref.~\cite{morgan_collapse_2022}.} With this, setting $\Kcal_i^\dagger[\bullet]:= \sum_{m_i} \Kcal^{(m_i)\dagger}_i[\bullet]$, Kolmogorov consistency is equivalent to 
\begin{gather}
\begin{split}
\label{eqn::LudersMulti}
    \, \text{tr}\left[\widetilde \rho_i(\mathbf{m}_{i-1:1}) \Kcal_i^\dagger[Q_{i}(\mathbf{m}_{n:i+1})]\right]  \\
    =\, \text{tr}\left[\widetilde \rho_i(\mathbf{m}_{i-1:1}) Q_{i}(\mathbf{m}_{n:i+1})\right], 
\end{split}
\end{gather}
which can be expressed as a commutator expression via
\begin{equation}
\label{eqn::LudersMulti-comm}
    \sum_{m_i}\text{tr}\left[\widetilde \rho_i(\mathbf{m}_{i-1:1}) K_i^{(m_i) \dagger}[K_i^{(m_i)},Q_{i}(\mathbf{m}_{n:i+1})]\right]=0 
\end{equation}
for all $t_i$ and all outcomes $\{m_1, \dots, m_{i-1}, m_{i+1}, \dots, m_n\}$. Formally, apart from the dependence on past and future outcomes, the above equation coincides with Eq.~\eqref{eq:luedersnoninvasiveness}, which states the Kolmogorov consistency conditions for two sequential measurements without intermediate dynamics. This seemingly implies that there should be a direct relation between commutation relations of the involved operators in Eq.~\eqref{eqn::LudersMulti-comm} and classicality of the observed statistics. However, we now discuss some important differences between classicality for two-time vis-a-vis multi-time processes and subsequently demonstrate that there is no `straightforward' extension of L{\"u}ders' theorem to the multi-time setting, except under rather restrictive assumptions.

\subsection{Two-time vs. Multi-time Classicality}

There are a number of major differences between the two- and multi-time scenario. Firstly, in the two-time case, assuming that measurement non-invasiveness holds for arbitrary initial states $\rho$, one can conclude that satisfaction of the Kolmogorov consistency conditions is equivalent to $\Kcal^\dagger[Q^{(b)}] = Q^{(b)}$. On the other hand, in the multi-time case, even when assuming that the Kolmogorov consistency conditions hold for arbitrary initial states (i.e., those at time $t_1$), one is no longer guaranteed that the system states span a full basis at every later time $t_i$. To see this, consider the natural case of measurements being performed in the computational basis, i.e., $\Kcal_{i-1}^{(m_{i-1})}[\rho] = \braket{m_{i-1}|\rho|m_{i-1}}\ketbra{m_{i-1}}{m_{i-1}}$. Then, the state $\widetilde{\rho}_i$ immediately prior to the measurement at time $t_i$---independent of that at the beginning of the experiment---is proportional to $\Lambda_{i:i-1}[\ketbra{m_{i-1}}{m_{i-1}}]$ and so the set of states $\{\widetilde \rho_i\}$ can at most span a $d$ dimensional space, which cannot coincide with the $d^2$ dimensional space spanned by all quantum states. Consequently, in the multi-time case, Kolmogorov consistency is---in contrast to the two-time scenario---manifestly \textit{not} equivalent to 
\begin{gather}
\label{eqn::FixPointMulti}
    \Kcal_i^\dagger[Q_{i}(\mathbf{m}_{n:i+1})] = Q_{i}(\mathbf{m}_{n:i+1})\, . 
\end{gather}
In particular, satisfaction of the above equation for all $i$ is sufficient for satisfaction of Kolmogorov conditions---as can be seen by direct insertion into Eq.~\eqref{eqn::LudersMulti}---but not necessary (see below). Nonetheless, the formulation of Eq.~\eqref{eqn::LudersMulti-comm} informs us that the commutation relations $[K_i^{(m_i)},Q_i(\mathbf{m}_{n:i+1})]$, or variants thereof, are the relevant ones to investigate with respect to satisfaction of Kolmogorov consistency conditions. 

Note that if \textit{all} CP maps $\{\Kcal_{i}^{(m_i)}\}$ and \textit{all} intermediate dynamics $\{\Lambda_{i:i-1}\}$ are invertible (the former is a choice, while the latter is generally true for Markovian dynamics~\cite{jeknic-dugic_invertibility_2021}), then one is guaranteed a full basis of states at each time $t_i$ provided that the very initial state is arbitrary; this implies that the Kolmogorov consistency conditions are then indeed equivalent to satisfaction of Eq.~\eqref{eqn::FixPointMulti}, even in the multi-time scenario. In this case, one can, in the spirit of Thm.~\ref{thm::Luders}, establish a direct connection between measurement non-invasiveness for all input states and the vanishing of the commutator $[K_i^{(m_i)},Q_i(\mathbf{m}_{n:i+1})]$ (provided some additional assumptions are met, e.g., Hermiticity of the Kraus operators $K_i^{(m_i)}$---we will revisit these assumptions and the issue of requiring a full basis in Sec.~\ref{subsec:classical-com}). However, with respect to the choice of instruments in particular, this restriction is rather strong and would (as mentioned above) fail to cover the most natural scenario of measurements in a fixed basis. 

Additionally, independent of the fact that at $t_i$ one does not have access to a full basis of states, when considering general quantum measurements in multi-time processes, information can be transmitted through the system alone, and thus measurement statistics can be correlated over multiple points in time, even for Markovian processes probed with an independent measurement sequence. Thus, in the multi-time scenario, one must deal with entire sequences of outcomes---for example, $Q_i(\mathbf{m}_{n:i+1})$ is an operator that pertains to the entire sequence of future outcomes---instead of just outcomes at single or neighbouring times. As we discuss in detail in Sec.~\ref{sec::NCGD}, this added complexity cannot be circumvented as soon as general measurements are considered, and consequently all of our results will be phrased with respect to operators that generally correspond to measurement outcomes at multiple different points in time. 

We now detail how Eq.~\eqref{eqn::FixPointMulti} has to be modified in order to yield a direct relation between commutativity of pertinent operators and the classicality of the observed statistics in a multi-time experiment.

\subsection{No `Straightforward' Extension of L{\"u}ders' Theorem}
\label{sec::no_ext_Lud}
In Eq.~\eqref{eqn::LudersMulti}, we have expressed satisfaction of Kolmogorov conditions at an arbitrary time $t_i$ in terms of the measurement map $\Kcal_i^\dagger$, the post-measurement operators $Q_i(\mathbf{m}_{n:i+1})$ and the pre-measurement subnormalised states $\widetilde \rho(\mathbf{m}_{i-1:1})$. The close formal relation of said equation to those that appear in the two-time scenario, and thus in L{\"u}ders' theorem, informs us that the commutators $[K_i^{(m_i)}, Q_i(\mathbf{m}_{n:i+1})]$ play a pivotal role for the classicality of the observed statistics. As mentioned, it is easy to see that $[K_i^{(m_i)}, Q_i(\mathbf{m}_{n:i+1})] = 0$ implies satisfaction of Kolmogorov conditions, but the converse is not true (see Ex.~\ref{ex::Counter} for a concrete counterexample). Indeed, such a commutation relation is, at the outset, far too strict a condition to be necessary for Kolmogorov consistency: It would, for instance, imply that the measurements are non-invasive for \textit{arbitrary} system states at each time, which is not a necessary requirement for classicality, since for many relevant scenarios, the possible states $\widetilde \rho_i$ before a measurement at time $t_i$ do not span the full space of quantum states. As a result, our aim is to find weaker commutation relations that still guarantee the satisfaction of Kolmogorov conditions, and, conversely, to work out the consequences of Kolmogorov conditions on the commutation relations of the relevant operators. Here, we begin with the former direction. 

Recall that satisfaction of Kolmogorov consistency conditions is given by Eq.~\eqref{eqn::LudersMulti}. As we have emphasised, this does not necessarily imply that $\Kcal^{\dagger}_i[Q_i] = Q_i$ and is thus not equivalent to $[K_i^{(m_i)},Q_i] =0$. Following the logic of Thm.~\ref{thm::Luders}, one might then suspect that 
\begin{equation}
\label{eq:Comm-procc}
	\tr{\widetilde{\rho}_i(\mathbf{m}_{i-1:1})~\![K_{i}^{(m_i)},Q_i(\mathbf{m}_{n:i+1})]} = 0\, ,
\end{equation}
i.e., commutativity of the measurement operators \textit{with respect to} $\widetilde{\rho}_i(\mathbf{m}_{i-1:1})$ is equivalent to the satisfaction of Kolmogorov conditions at time $t_i$---at least for the case of Hermitian Kraus operators $K_{i}^{(m_i)}$. However, this is not the case, as the following example shows:
\begin{example}
\label{ex::non-Luders} 
Let the pre-measurement state (for some history of outcomes, which we renormalise and suppress for the sake of conciseness) be given by $\widetilde{\rho}_i =\dfrac{1}{2}\cv{\iden + \sigma_z}$ followed by a measurement described by Kraus operators $K_{i}^{(\pm)}=\dfrac{1}{2}\cv{\iden \pm \sigma_x}$ and let the post-measurement part (for some sequence of future outcomes) be encoded in the operator $Q_i\cv{\pm} =\dfrac{1}{2}\cv{\iden \pm \sigma_z}$. This situation can, e.g., arise in a two-step process without intermediate evolution, where the measurement with Kraus operators $\{K_{1}^{(\pm)}\}$ is made at $t_1$, the pre-measurement state is prepared as $\widetilde{\rho}_{1}=\dfrac{1}{2}\cv{\iden + \sigma_z}$ and at $t_2$ the observable $\sigma_z$ is measured with outcomes $\pm$, corresponding to the post-measurement operators $Q_1\cv{\pm} = \dfrac{1}{2}\cv{\iden \pm \sigma_z}$. We observe that (for future outcome $+$) we have $\tr{\widetilde{\rho}_1\cvb{K_{1}^{(\pm)},Q_1\cv{+}}} = 0$ but 
\begin{equation}
\begin{split}
		&\tr{\widetilde{\rho}_1 \Big(\sum_{m_1=\pm}K_{1}^{(m_1) \dagger}Q_1\cv{+} K_{1}^{(m_1)}-Q_1\cv{+}\Big)} \\
		&=  \tr{\widetilde{\rho}_1 \Big(\Kcal_{1}^\dagger[Q_1\cv{+}] -Q_1\cv{+}\Big)} = \frac{1}{2} \neq 0 \, ,
\end{split}
\end{equation}
implying that the Kolmogorov consistency conditions are not satisfied even though commutativity with respect to $\widetilde \rho_1$ [i.e., Eq.~\eqref{eq:Comm-procc}] holds and all involved Kraus operators are Hermitian. \hfill $\blacksquare$
\end{example} 

We emphasise that even though we only explicitly consider two measurements here, the considered scenario is indeed a multi-time one; in contrast to the scenario envisioned by L\"uders, we do not assume the states before the first measurement at $t_1$ to span a full basis, which is an implicit assumption of the two-time setting with arbitrary initial preparations. This, in turn, can be understood as the pre-measurement state $\widetilde \rho_1$ being the result of a previous measurement with a fixed outcome (or sequence thereof), making the scenario of the example a genuine multi-time one of which we only explicitly investigated the two times $t_1$ and $t_2$. The fact that Kolmogorov conditions are not satisfied despite the weak commutativity of Eq.~\eqref{eq:Comm-procc} holding then signifies that in the multi-time scenario, one requires a stricter commutation relation for the involved operators in order to obtain classical statistics.

Although the weak commutation relation with respect to $\widetilde{\rho}_i(\mathbf{m}_{i-1:1})$ is not restrictive enough, the following theorem informs us that \textit{absolute} commutativity with respect to $\widetilde{\rho}_i(\mathbf{m}_{i-1:1})$ is indeed sufficient to guarantee the satisfaction of Kolmogorov conditions: 

\begin{thm}\label{thm:abs-CM-to-KM}
Let $\mathbbm{P}(m_n,\hdots,m_1|\Jcal_n, \hdots, \Jcal_1)$ be a joint probability distribution obtained from Eq.~\eqref{eq:qrf}, i.e., by probing a Markovian process. If absolute commutativity 
\begin{gather}
\label{eqn::AbsComm}
    \tr{\widetilde{\rho}_{i}(\mathbf{m}_{i-1:1}) \,\big| \, [K_i^{(m_i)}, Q_i(\mathbf{m}_{n:i+1})]\, \big|} = 0
\end{gather} 
holds at all times $t_i$ and for all possible $\textbf{m}_{i-1:1}$, $m_i$ and  $\textbf{m}_{n:i+1}$, where $|X| := \sqrt{X^\dagger X}$, then $\mathbbm{P}(m_n,\hdots,m_1|\Jcal_n, \hdots, \Jcal_1)$ satisfies the Kolmogorov consistency conditions \textup{[}given explicitly in Eq.~\eqref{eqn::LudersMulti-comm}\textup{]}. 
\end{thm}
\begin{proof} 
For simplicity, we will omit the explicit arguments of the involved operators throughout the proof. We first show that Eq.~\eqref{eqn::AbsComm} implies $\tr{\widetilde{\rho}_{i}~K_{i}^{(m_i) \dagger}[K_{i}^{(m_i)},Q_i]} = 0$. To this end, we note that Eq.~\eqref{eqn::AbsComm} implies 
\begin{gather}
     \,\big| \, [K_i^{(m_i)}, Q_i]\, \big| ~\widetilde{\rho}_{i} = 0\, ,
\end{gather}
since both $\widetilde{\rho}_{i}$ and $\big| \, [K_i^{(m_i)}, Q_i]\, \big|$ are positive semidefinite. Now, let us employ the polar decomposition $[K_i^{(m_i)}, Q_i] = V^{(m_i)} M^{(m_i)}$, where $V^{(m_i)}$ is unitary and $M^{(m_i)} = \big| \, [K_i^{(m_i)}, Q_i]\, \big| \geq 0$, i.e., $M^{(m_i)}  \widetilde{\rho}_{i}$ = 0. With this, we obtain 
\begin{gather}
\begin{split}
    &\tr{\widetilde{\rho}_{i}~K_{i}^{(m_i) \dagger}[K_{i}^{(m_i)},Q_i]} \\
    &\phantom{as}= \tr{K_{i}^{(m_i) \dagger}V^{(m_i)} M^{(m_i)} \widetilde{\rho}_{i}} = 0\, .
\end{split}
\end{gather} 
By summing this expression over $m_i$, Eq.~\eqref{eqn::LudersMulti-comm}---and thus satisfaction of the Kolmogorov consistency conditions---is recovered. \end{proof}

Thm.~\ref{thm:abs-CM-to-KM} informs us that absolute commutativity with respect to the state of the system at each time is sufficient for classicality of the observed statistics. However, in contrast to the two-time scenario, this requirement is \textit{not} necessary for classicality---even in the case where the involved Kraus operators are Hermitian. To see this, consider the following example: 
\begin{example}
\label{ex::Counter}
We employ Ex.~\ref{ex::non-Luders} with a change in the (renormalised) pre-measurement part to $\widetilde{\rho}_{1}=\dfrac{1}{2}\cv{\iden + \sigma_y}$ but still followed by a measurement $K_{1}^{(\pm)}=\dfrac{1}{2}\cv{\iden \pm \sigma_x}$ and the post-measurement parts are encoded in the operators $Q_{1}\cv{\pm}=\dfrac{1}{2}\cv{\iden \pm\sigma_z}.$ 
We observe that---for these choices---Kolmogorov consistency holds, i.e., 
	\begin{equation*}
		\tra\Big\{\widetilde\rho_1 \Big[\sum_{m_1=\pm}K_{1}^{(m_1) \dagger}Q_1\cv{\pm} K_{1}^{(m_1)}-Q_1\cv{\pm}\Big]\Big\} = 0 \, ,
	\end{equation*}
since $\sum_{m_1=\pm}K_{1}^{(m_1) \dagger}Q_1\cv{\pm} K_{1}^{(m_1)}-Q_1\cv{\pm} =\pm\dfrac{\sigma_z}{2}$, which is trace orthogonal to $\widetilde{\rho}_1$. However we find that $\big\vert\cvb{K_{1}^{(\pm)},Q_1\cv{\pm}}\big\vert=\dfrac{\iden}{4}$ leading to 
	\begin{equation*}
		\tr{\widetilde \rho_1\big\vert\cvb{K_{1}^{(\pm)},Q_1\cv{\pm}}\big\vert} = \frac{1}{4} \neq 0 \, .
	\end{equation*}
Consequently, this example shows that classical statistics in a multi-time experiment do not imply absolute commutativity with respect to the state of the interrogated system over time, even when all Kraus operators are Hermitian (which is the case here). In turn, since absolute commutation with respect to $\widetilde \rho_i$ is weaker than commutativity itself, this makes the considered case also an example of a situation where satisfaction of Kolmogorov consistency conditions does not imply $\cvb{K_{i}^{(m_i)},Q_i(\textbf{m}_{n:{i+1}})} = 0$, as mentioned at the beginning of this section. \hfill $\blacksquare$
\end{example}

While not being equivalent to satisfaction of the Kolmogorov consistency conditions, absolute commutativity with respect to the state $\widetilde \rho_i$ guarantees classical statistics and, in contrast to the much stronger standard commutativity condition, does not necessarily imply Kolmogorov consistency independent of the sequentially measured system states, making it a more relevant consideration for the envisaged scenario. 

Regarding this connection between commutativity and classicality, two remarks are in order. On the one hand, if $\widetilde \rho_i$ is \textit{full rank}, then it is easy to see that $\tr{\widetilde{\rho}_{i}(\mathbf{m}_{i-1:1}) \big|[K_i^{(m_i)}, Q_i(\mathbf{m}_{n:i+1})]\big|} = 0$ implies $[K_i^{(m_i)}, Q_i(\mathbf{m}_{n:i+1})] = 0$, thus equating the assumption of absolute commutativity with respect to $\widetilde \rho_i$ to the (rather strong) assumption of standard commutativity. However the states $\widetilde{\rho}_{i}$ do not necessarily have to be full rank. This holds true, e.g., for the case of pure projective measurements in a fixed basis on a qutrit and intermediate dynamics that only map to the (lower dimensional) space that is spanned by $\{\ket{0}, \ket{1}\}$. In turn, this makes the assumption of Thm.~\ref{thm:abs-CM-to-KM} strictly weaker than full commutativity, while still being strictly stronger than commutativity with respect to $\widetilde \rho_i$ [i.e., satisfaction of Eq.~\eqref{eq:Comm-procc}], which, as we have seen, is \textit{not} sufficient to guarantee classical statistics.
	
On the other hand, the states $\widetilde \rho_i$ at time $t_i$ are the result of a state preparation at the initial time, followed by a sequence of measurements and intermediate dynamics. Assuming that a full basis of initial states can be prepared (as is assumed in the two-time scenario envisioned by L{\"u}ders), then it is---in principle---possible that, for each sequence of outcomes, the corresponding states $\widetilde \rho_i$ also span a basis at each time $t_i$. In this case, satisfaction of Kolmogorov conditions at $t_i$ would exactly coincide with $\Kcal_i^\dagger[Q_i(\textbf{m}_{n:i+1})] = Q_i(\textbf{m}_{n:i+1})$ [see Eq.~\eqref{eqn::LudersMulti}] and, following the same reasoning that led to L{\"u}ders' Thm.~\ref{thm::Luders}, we would be able to recover the equivalence between classical statistics and the vanishing of the commutators $[K_i^{(m_i)}, Q_i(\textbf{m}_{n:i+1})] = 0$. In this sense, it might seem artificial to investigate the case where states at each time do \textit{not} span a full basis, which, as we will see, leads to a more layered relationship between commutativity and classicality. However, this latter case exactly mirrors many physically relevant scenarios (like, e.g., the case of sequential projective measurements).

This inequivalence between commutation relations and classicality naturally raises the question: \emph{What further assumptions, in addition to classicality of statistics, must be satisfied in order to ensure commutation relations of the relevant operators?} 

\subsection{Commutativity as a Notion of Classicality: When is Kolmogorov Consistency Sufficient for L{\"u}ders-type Theorems?}\label{subsec:classical-com} 

In this section we investigate under which conditions satisfaction of Kolmogorov consistency implies the vanishing of pertinent commutator expressions. Unlike the previous sections, here---just like in the scenario considered by L\"uders---we have to restrict the Kraus operators of the probing instruments to be Hermitian in order to establish a clear connection between Kolmogorov consistency and vanishing commutators. 

As mentioned previously, a key element that makes the multi-time setting substantially different to the two-time one is that one is no longer guaranteed a full basis of quantum states at each time. Nonetheless, below we outline a condition that ensures that the set of possible states at each time (conditioned on previous outcome sequences) essentially forms a basis \textit{with respect to} any subsequent measurements (additionally accounting for the intermediate dynamics). Analogously to the case considered by L\"uders, our argument requires Hermiticity of the measurement Kraus operators; we leave the analysis of sufficient conditions regarding more general measurements in this setting for future work. The conditions that we detail below consequently ensures a connection between Kolmogorov consistency and commutativity in the multi-time setting (for Hermitian Kraus operators). Importantly, just like in the case of L\"uders, under the additional assumptions we make, commutativity and classicality are equivalent.

To establish the connection between Kolmogorov consistency and commutation relations, consider the set of possible pre-measurement states at some time (say, $t_i$) of interest. As mentioned, we can follow a L\"uders type argument, if these states form a basis with respect to the post-measurement operators $Q_i$. Formally, we can express this by letting $\mathbb{S}$ be a set of initial states $\rho$ and $\mathbb{H}_i$ be the span of the union of the images of all possible pre-measurement sequences up until time $t_i$:
\begin{gather}
\begin{split}
    &\mathbb{H}_i := \mathrm{span}\\
    &\bigcup_{\mathbf{m}_{i-1:1}}\sum_{\ell_2\dots \ell_i} R^{(\ell_i,m_{i-1})}_{i:i-1} \cdots R^{(\ell_2,m_1)}_{2:1} \mathbb{S} R^{(\ell_2,m_1)\dagger}_{2:1} \cdots R^{(\ell_i,m_{i-1})\dagger}_{i:i-1} \, , \label{eq:H-def}
\end{split}
\end{gather}
for $i\geqslant 2$, i.e., $\mathbb{H}_i$ is the span of all attainable states $\widetilde \rho_i(\mathbf{m}_{i-1:1})$ at time $t_i$. Furthermore, we take the union of all possible projections for the post-measurement operators to define:
\begin{align}
    \mathbb{F}_i &:= \mathrm{span}\bigcup_{\mathbf{m}_{n:i-1}}\bigg\{P_\mu : Q_i\cv{\mathbf{m}_{n:i-1}}=\sum_\mu\lambda_\mu P_\mu\bigg\}, \label{eq:F-def}
\end{align}
where, for technical reasons, we will assume non-degeneracy of $Q_i$ (see the proof of Thm.~\ref{thm:KM-to-abs-CM}). Demanding that the pre-measurement states form a basis with respect to the post-measurement operators now amounts to the requirement $\mathbbm{F}_i \subseteq \mathbbm{H}_i$. As it turns out, together with the satisfaction of Kolmogorov consistency, this implies that  $\mathrm{tr}[P_\mu\mathcal{K}_i^\dagger[Q_i]] = \mathrm{tr}[P_\mu Q_i] \, \forall \, \mu$, which suffices to prove that the pertinent commutation relations hold (under the assumption that all Kraus operators pertaining to the measurement map $\Kcal^\dagger$ are Hermitian):

\begin{thm}\label{thm:KM-to-abs-CM}
Let $\mathbbm{P}(m_n,\hdots,m_1|\Jcal_n, \hdots, \Jcal_1)$ be a joint probability distribution obtained from Eq.~\eqref{eq:qrf}, i.e., by probing a Markovian process. Assume that $\mathbbm{P}(m_n,\hdots,m_1|\Jcal_n, \hdots, \Jcal_1)$ satisfies the Kolmogorov consistency conditions \textup{[}given explicitly in Eq.~\eqref{eqn::LudersMulti-comm}\textup{]} for all initial state $\rho$ in $\mathbb{S}$ and for every measurement time $t_i$ and that $Q_i(\mathbf{m}_{n:i+1})$ is non-degenerate for all $\textbf{m}_{n:i+1}$. If all Kraus operators $K_i^{(m_i)}$ are Hermitian for all $m_i$ and
    \begin{align}\label{eq:F-sub-H}
         \mathbb{F}_i&\subseteq\mathbb{H}_i ,
    \end{align}
then the commutation relations hold, i.e., 
\begin{equation}
        \cvb{K_i^{(m_i)}, Q_i(\mathbf{m}_{n:i+1})} = 0 \label{eq:CM-thm4}
\end{equation}
for all post-measurement sequences $\mathbf{m}_{n:i+1}$ and all $m_i$.
\end{thm}
Before providing the proof of Thm.~\ref{thm:KM-to-abs-CM}, we emphasise that the converse trivially holds (even without any assumptions), since commutativity [i.e., Eq.~\eqref{eq:CM-thm4}] directly implies the satisfaction of Kolmogorov consistency. 
\begin{proof}
From the assumption Eq.~\eqref{eq:F-sub-H}, one can see that for a given post-measurement sequence $\textbf{m}_{n:i+1}$ and for any pre-measurement sequence $\textbf{m}_{i-1:1}$, there exists an initial state $\rho$ in $\mathbb{S}$ leading to $\widetilde{\rho}_{i}(\mathbf{m}_{i-1:1})=P_\mu$ for any $P_\mu$ defined via $Q_i\cv{\mathbf{m}_{n:i-1}}=\sum_\mu\lambda_\mu P_\mu$. In other words, Kolmogorov consistency in the form of Eq.~\eqref{eqn::LudersMulti} leads to $\mathrm{tr}[P_\mu\mathcal{K}_i^\dagger[Q_i]] = \mathrm{tr}[P_\mu Q_i] \, \forall \, \mu$.

Now, using the same arguments as those of the proof of Thm.~\ref{thm::Luders}, one sees that $\text{tr}[P_\mu\mathcal{K}_i^\dagger[Q_i]] = \text{tr}[P_\mu Q_i]$ for $\mu=1$ leads to 
	\begin{gather}
	\label{eq:Luder-train_identity-mod}
	\begin{split}
		&\sum_m\cvr{K^{(m_i)}_i\varphi_1\, \vline\left(\lambda_1\iden-Q_i\right)K^{(m_i)}_i\varphi_1}  \\
		&\phantom{asdfasdf} =\sum_m\cvV{\left(\lambda_1\iden-Q_i\right)^{1/2}K^{(m_i)}_i\ket{\varphi_1}}^2 = 0 \, ,
	\end{split}
	\end{gather}
where we set $P_1=:\ketbra{\varphi_1}{\varphi_1}$ and made use of the fact that $\lambda_1\iden-Q_i\geqslant 0.$ In other words, $\left(\lambda_1\iden-Q_i\right)^{1/2}K^{(m_i)}_iP_1=0$ or
	\begin{equation}
		Q_iK_i^{(m_i)}P_1\ket{\varphi} = \lambda_1K^{(m_i)}P_1\ket{\varphi} 
	\end{equation}
for arbitrary states $\ket{\varphi},$ i.e., $K_i^{(m_i)}$ leaves the $\lambda_1-$eigensubspace invariant. Thus, due to non-degeneracy of $Q_i$, we have $K^{(m_i)}_iP_1=P_1 K^{(m_i)}_iP_1$ and assuming that the $K^{(m_i)}_i$ are Hermitian, we observe that $[K^{(m_i)}_i,P_1]=0.$ Again, we set $Q_i^{(\mu)}=Q_i^{({\mu-1})}-\lambda_{\mu-1} P_{\mu-1}$ and $Q_i^{(0)}=Q_i.$ Since $[K^{(m_i)}_i,P_1]=0,$ the expression $\mathrm{tr}[P_2\mathcal{K}_i^\dagger[Q_i]] = \mathrm{tr}[P_2 Q_i]$ can be reduced to $\mathrm{tr}[P_2\mathcal{K}_i^\dagger[Q^{(2)}_i]] = \mathrm{tr}[P_2 Q^{(2)}_i]$ and then it follows that $[K^{(m_i)}_i,P_2]=0$ (by invoking the same previous argument but replacing $Q_i$ and $P_1$ with $Q^{(2)}_i$ and $P_2$, respectively). Iterating this argument---as in the proof of Thm.~\ref{thm::Luders}---we obtain that $K^{(m_i)}_iP_\mu=P_\mu K^{(m_i)}_iP_\mu$ for all $m$ and $\mu,$ i.e., $K^{(m_i)}_i$ leaves all eigensubspaces of $Q$ invariant. Then $\cvb{K_i^{(m_i)}, Q_i(\mathbf{m}_{n:i+1})} = 0$ for all post-measurement sequences $\mathbf{m}_{n:i+1}$ and all $m_i$ as claimed.
\end{proof}
We emphasise that---as in the case of L\"uders' theorem---this logic can fail to hold if the $K^{(m_i)}_i$ are not Hermitian (as can already be explicitly seen by considering Ex.~\ref{eqn::ExNonHerm}).

For illustration of the above theorem, let us consider the following example.

\begin{example}
\label{ex::thm4-eg}
Recall the scenario of Ex.~\ref{ex::non-Luders}. We modify it to be a three step process with measurements in the $\sigma_z$-basis at the first and third time, while at the second time, a measurement in the $\sigma_x$-bais is carried out. In addition, let the dynamics between the first and the second measurement, as well as between the second and the third measurement, be given by a Hadamard gate $H$, with $H\ket{0/1} = \ket{\pm}$ and $\ket{\pm} = 1/\sqrt{2}(\ket{0} \pm \ket{1})$. 

We focus on time $t_2$ as the measurement time of interest (i.e., the time for which we analyse Kolmogorov consistency). For an arbitrary initial state $\rho$, the measurement at $t_1$ leads to the set of possible post-measurement states $\widetilde{\rho}_1 \in \{\ketbra{0}{0}, \ketbra{1}{1}\}$, where we omit potential subnormalisation. These states will then evolve to $H\cvb{\widetilde{\rho}_1}H=\widetilde{\rho}_2 \in \{ \ketbra{+}{+}, \ketbra{-}{-} \}$. Thus, at time $t_2$, we have $\mathbbm{H}_2 = \mathrm{span} \{ \ketbra{+}{+}, \ketbra{-}{-}\}$.

Since a measurement in the $\sigma_z$-basis is performed at $t_3$, the post-measurement part at $t_2$ amounts to the corresponding projectors, `rolled back' by means of the evolution $\Lambda^\dagger_{3:2}[\bullet] = H[\bullet]H$, i.e., we have $Q_{2}\cv{0/1} \in \{ \ketbra{+}{+}, \ketbra{-}{-}\}$. With this, we observe that 
    \[\mathbb{F}_2=\mathbb{H}_2\]
and thus the condition of Eq.~\eqref{eq:F-sub-H} holds. Likewise, for the measurement $K_{2}^{(\pm)}$, the Kolmogorov consistency condition reads
 \begin{equation}
     \tr{\widetilde{\rho}_2\bigg(\sum_mK^{(m) \dagger}_2\cvb{K^{(m)}_2,Q_2\cv{0/1}}\bigg)} = 0.
 \end{equation}    
where $m \in\{+,-\}$. We can calculate the commutativity expression of Eq.~\eqref{eq:CM-thm4} explicitly:
\begin{gather}
    \cvb{K^{(\pm)}_2,Q_2\cv{0/1}}= \cvb{\ketbra{\pm}{\pm},\ketbra{\pm}{\pm}}=0.
\label{eqn::ProjComEx}
\end{gather}\hfill $\blacksquare$
\end{example}

The inclusion property of Eq.~\eqref{eq:F-sub-H} is a rather strong requirement, which---as we have seen in Ex.~\ref{ex::thm4-eg}---can be checked for and satisfied in particular cases. However, it can fail to hold for many experimentally relevant situations that yield classical statistics (like, e.g., measurements in a fixed basis, see Ex.~\ref{ex::NCGD_incl} below). Additionally, one would ideally like to deduce similar conditions that apply to arbitrary (e.g., non-Hermitian) measurements at the expense of potentially weakening the vanishing commutator expression; so far, such results have proved elusive. As a consequence, in the multi-time setting, the relation between observed classicality and the commutation of pertinent operators presents itself much more layered than in the two-time case considered by L\"uders, and must seemingly be decided on a case-by-case basis. 

We now finish our discussion of Markovian classical multi-time processes by discussing why, even though the underlying process is memoryless, it is, in general, necessary to consider the \textit{full} history (future) of outcomes $\mathbf{m}_{i-1:1}$ ($\mathbf{m}_{n:i+1}$) at each time $t_i$, and not just the preceding (subsequent) ones $m_{i-1}$ ($m_{i+1}$). The latter (i.e., only considering outcomes at $t_{i-1}$ and $t_{i+1}$ for each time $t_i$) can be done for the special case of projective pure measurements in a fixed basis, which leads to an equivalent formulation of classicality in terms of non-coherence-generating-and-detecting maps. However, as we will see in the following section, a direct connection between such dynamics to pertinent commutation relations is generally not obvious.

\subsection{Markovian Processes \& NCGD Dynamics}
\label{sec::NCGD}

Up to this point, we have investigated the conditions under which Kolmogorov consistency of observed statistics and the vanishing of pertinent commutator expressions---i.e., structural properties pertaining to the dynamics and measurement scheme---are related in the multi-time setting. Crucially, we see that Kolmogorov consistency concerns a deep interplay between the choice (and assumptions) of measurements and the underlying dynamics that depends on entire sequences of measurement outcomes. 

On the other hand, for the Markovian case we consider, recent work has demonstrated a one-to-one connection between Kolmogorov consistency and the (in)ability for pairs of \textit{neighbouring} dynamical maps $\Lambda_{i:i-1}$ (describing the open evolution of Markovian process) to generate and detect coherence with respect to a fixed basis determined by the measurement scheme~\cite{Smirne2019b, Strasberg2019, Milz2019,  Smirne2019}. Specifically, these works showed equivalence between classical statistics and the set of so-called non-coherence-generating-and-detecting (\textbf{NCGD}) dynamics, i.e., maps that can create coherences, but those coherences cannot be `detected' by the subsequent dynamics [see Eq.~\eqref{eqn::NCGD_rephrasing} for a proper definition]. While this criterion can be phrased entirely in terms of dynamical maps pertaining to neighbouring times, our work has required the consideration of \textit{entire measurement sequences} of past $\mathbf{m}_{i-1:1}$ and future $\mathbf{m}_{n:i+1}$ outcomes in general, instead of simply adjacent ones.  

At the outset, this necessity seems to be overkill, since, intuitively, the statistics measured at each time of a \emph{Markovian} process should only depend upon the most recent outcome, and not on the entire history. We now return to elucidate why, even though the underlying dynamics that we study are assumed to be Markovian, one must indeed consider commutator expressions of the relevant operators corresponding to entire sequences. The important subtlety to note here is that general quantum measurements \emph{do not} break the flow of information through the measured system, and thus even though there is no \emph{non-Markovian memory} (travelling through an environment), the observed statistics can still be correlated over multiple times. Put differently, after a general measurement, the state of the measured system is unknown, and might depend on earlier measurements, even though the underlying process itself does not exhibit any non-Markovian memory. Here and in what follows, by `flow of information', we mean any dependence of the post-measurement state on the measured state, and thus on previous measurement outcomes/choices of instruments. For instance, a sharp, projective measurement with outcome $m$ sends $\rho$ to $\ketbra{m}{m}$ and thus breaks the flow of information since the post-measurement state only depends on the observed outcome $m$, but not on any previous measurements or manipulations of the measured state $\rho$; on the other hand, a POVM with elements $\{ \xi^{(m)}\}$ that sends $\rho \mapsto \xi^{(m)}\, \rho \,\xi^{(m) \dagger}$ does not break the flow of information, since the post-measurement state (i.e., not only the probability to observe the measurement outcome $m$) still depends on $\rho$ and in general any previous manipulations thereof. We emphasize that, throughout, we uniquely use the term `flow of information' in this sense, with no concrete reference to technical notions of the term `information' employed in quantum information science.

Importantly, this point is not critically related to any inherently `quantum' notion regarding the measurement (such as being a POVM comprising non-projective, non-Hermitian, or non-orthogonal elements), but can occur for \textit{any} measurement for which an outcome does not fully determine the post-measurement state of the system. This can happen in classical physics for `fuzzy' measurements that coarse grain over different levels~\cite{Taranto2019A}, and is generally the case for measurements in quantum mechanics described by CP maps that do \textit{not} necessarily break the information flow through the system. For such measurements, the far past can still have an influence on the future~\cite{Pollock2018a, Pollock2018b,Taranto2019L,Taranto2019A,Taranto2019S,TarantoThesis}. As a result, measurement invasiveness might not be detected at the next step, but possibly only further in the future, and any conditions pertaining to neighbouring dynamical maps alone are insufficient to characterise classicality. To see this explicitly, consider the following example, which concerns noisy (i.e., not rank $1$) orthogonal measurements:
\begin{example}
Let $\rho$ be the state of a four level system that is measured at times $\{t_1, t_2, t_3\}$ by means of projective---but not rank-1---measurements, i.e., $\Jcal_i = \{K_i^{(1)} = \Pi_{(12)}, K_i^{(2)} = \Pi_{(34)}\}$, where $\Pi_{(xy)}$ is the projector on the space spanned by $\{\ket{x},\ket{y}\}$, with $x,y \in \{1,2,3,4\}$. Now, let the dynamics in between measurements be given by the Kraus operators $\{L_{2:1}^{(1)} = \ketbra{1}{1} + \ketbra{2}{4}, L_{2:1}^{(2)} = \ketbra{2}{2} + \ketbra{4}{3}\}$ and $\{L_{3:2}^{(1)} = \tfrac{1}{\sqrt{2}}(\ketbra{3}{1} + \ketbra{3}{2}), L_{3:2}^{(2)} = \tfrac{1}{2}(\ketbra{1}{1} - \ketbra{1}{2} - \ketbra{2}{1} + \ketbra{2}{2}) + \ketbra{3}{3}
+ \ketbra{4}{4}\}$, respectively. These choices of Kraus operators correspond to CPTP maps, since $\sum_{\ell_{i+1}} L_{i+1:i}^{(\ell_{i+1})\dagger}L_{i+1:i}^{(\ell_{i+1})} = \ident$ for $i\in \{1,2\}$. It is easy to see that in this case, the statistics of the measurement at $t_2$ is independent of whether or not the measurement at $t_1$ was performed. Overall, the measurement at $t_1$ reduces the initial state $\rho$ to a block diagonal structure; however, the statistics at $t_2$ only depend on the diagonal terms of $\rho$, such that the invasiveness of the first measurement is not detected. Specifically, we have 
\begin{align}
    &\Pprob(m_2=1,\cancel{m_1}) = \rho_{11} + \rho_{22}+\rho_{44} = \sum_{m_1} \Pprob(m_2=1,m_1)\,, \notag \\
    &\Pprob(m_2=2,\cancel{m_1}) = \rho_{33} = \sum_{m_1} \Pprob(m_2=2,m_1) \,.
\end{align}
As a result, the two-time statistics do not reveal the non-classicality of the observed statistics, despite the dynamics being Markovian. However, the invasiveness of the first measurement can be observed via the measurement at time $t_3$. Concretely, the dynamics between $t_2$ and $t_3$ is such that it maps off-diagonal terms to diagonal ones, and thus the joint probability to measure $m_2=1$ and $m_3 = 1$ at times $t_2$ and $t_3$ (with \textit{no} measurement at $t_1$), respectively, is given by
\begin{gather}
    \Pprob(m_3=1,m_2=1,\cancel{m_1}) = \frac{1}{2} \left(\rho_{11} - 2\text{Re}(\rho_{14}) + \rho_{22}  + \rho_{44}\right).
\end{gather}
Since the above probability depends on the entry $\rho_{14}$ of the initial state $\rho$, it cannot coincide with the corresponding probability for the case where a measurement was performed at $t_1$. As mentioned, the overall action of the measurement at $t_1$ is to force $\rho$ into a block-diagonal structure, implying in particular $\rho_{14}\mapsto 0$ if a measurement at $t_1$ is performed. Consequently, we have 
\begin{align}
    \sum_{m_1} \Pprob(m_3=1,m_2=1,m_1) &= \frac{1}{2} (\rho_{11} + \rho_{22} + \rho_{44}) \notag \\
    &\neq \Pprob(m_3=1,m_2=1,\cancel{m_1}) \, . 
\end{align}
Accordingly, for the case of general instruments, one indeed must consider the full past and full future statistics for the relevant commutation relations in order to characterise classicality. \hfill $\blacksquare$
\end{example}

In the example above, we see that the invasiveness of the first measurement `skips' a time, i.e., it is not detected at time $t_2$ but rather only by the measurement at time $t_3$. Such `skipping' of detectability is not limited to measurement invasiveness and has recently been analysed with respect to the activation of hidden quantum memory~\cite{Taranto_2022}. 

Such behaviour highlights the intricacies involved when considering quantum processes probed sequentially at multiple times by general instruments. However, for particular types of measurements, the flow of information through the system is broken, and one can therefore connect classicality to structural properties of the underlying dynamical maps between only \emph{adjacent} times. This is, for example, the case if all measurements are rank-$1$ projective measurements in a fixed basis. Then, it is easy to see that $\widetilde{\rho}_i(\mathbf{m}_{i-1:1}) \propto \Lambda_{i:i-1}[\ketbra{m_{i-1}}{m_{i-1}}]$ and $Q_i(\mathbf{m}_{n:i+1}) \propto \Lambda_{i+1:i}^\dagger[\ketbra{m_{i+1}}{m_{i+1}}]$, where $\Lambda_{i+1:i}^\dagger[\bullet] = \sum_{\ell_{i+1}} L_{i+1:i}^{(\ell_{i+1})\dagger} \bullet L_{i+1:i}^{(\ell_{i+1})}$. With this, Eq.~\eqref{eqn::LudersMulti} reduces to 
\begin{align}
\label{eqn::NCGD_rephrasing}
    &\braket{m_{i+1}| \Lambda_{i+1:i} \circ \Delta_i \circ \Lambda_{i:i-1}[\ketbra{m_{i-1}}{m_{i-1}}]|m_{i+1}} \notag \\
    = &\braket{m_{i+1}| \Lambda_{i+1:i} \circ \Lambda_{i:i-1}[\ketbra{m_{i-1}}{m_{i-1}}]|m_{i+1}}
\end{align}
for all $t_i$ and all $\{m_{i-1}, m_{i+1}\}$, where $\Delta_i[\rho] := \sum_{m_i} \braket{m_i|\rho|m_i}\ketbra{m_i}{m_i}$ is the completely dephasing map at time $t_i$. Notably, Eq.~\eqref{eqn::NCGD_rephrasing} rephrases satisfaction of the Kolmogorov consistency conditions in terms of the properties of adjacent dynamical maps $\{\Lambda_{i+1:i}, \Lambda_{i:i-1}\}$ only, thus allowing for a full characterisation of Markovian dynamics that yield classical statistics when probed in a fixed basis, as is provided in Ref.~\cite{Smirne2019}. However, this is only possible since for rank-$1$ projective measurements, the state of the system after measurement is known (up to normalisation). Any measurement with this property breaks the information flow through the system, in the sense that, upon observing a given outcome, the future outcome statistics of a Markovian process cannot depend on any previous outcomes, since the state of the system has been completely reset~\cite{Pollock2018a,Pollock2018b,Taranto_2022}. This, in turn, is what allows one to characterise the classicality of Markovian processes in terms of neighbouring dynamical maps only, as per Eq.~\eqref{eqn::NCGD_rephrasing}. The above example shows why, in the case of \textit{general} measurements, one must consider operators corresponding to the entire future and history when discussing classicality, even for Markovian processes, as we have done throughout this article.

Finally, given that for the special case of (rank-$1$) projective measurements, NCGD dynamics provides a necessary \textit{and} sufficient condition for the classicality of the observed statistics, and all Kraus operators of the measurements are Hermitian, one might expect that L{\"u}ders-type assertions can be made with respect to commutation relations of pertinent operators (like those of Thms.~\ref{thm:abs-CM-to-KM} and~\ref{thm:KM-to-abs-CM}). However, this is not the case, as the following example demonstrates: 
\begin{example}
\label{ex::NCGD_incl}
Consider a three step qubit process on times $\{t_1,t_2,t_3\}$ with measurements in the computational basis and intermediate dynamics given by the CPTP maps
\begin{gather}
\begin{split}
    &\Lambda_{2:1}[\,\bullet\,] = \frac{1}{2}\left(\begin{array}{cc} 1 & 1 \\ \iu & -\iu \end{array}\right) \bullet \left(\begin{array}{cc} 1 & -\iu \\ 1 & \iu \end{array}\right) \\ 
    \text{and} \quad &\Lambda_{3:2}[\,\bullet\,] = \frac{1}{2}\left(\begin{array}{cc} 1 & 1 \\ 1 & -1 \end{array}\right) \bullet \left(\begin{array}{cc} 1 & 1 \\ 1 & -1 \end{array}\right)\, ,
\end{split}
\end{gather}
i.e., a rotation from the computational basis to the eigenbasis of $\sigma_y$ between $t_1$ and $t_2$, followed by a Hadamard gate between $t_2$ and $t_3$. It is easy to see that, for measurements in the computational basis 
\begin{gather}
\begin{split}
    &\braket{m_{3}| \Lambda_{3:2} \circ \Delta_2 \circ \Lambda_{2:1}[\ketbra{m_1}{m_1}]|m_3} \notag \\
    = &\braket{m_{3}| \Lambda_{3:2} \circ \Lambda_{2:1}[\ketbra{m_{1}}{m_{1}}]|m_{3}} = \frac{1}{2}
\end{split}
\end{gather}
for all $m_1, m_3 \in \{0,1\}$, and thus the dynamics is NCGD [since it satisfies Eq.~\eqref{eqn::NCGD_rephrasing}]. However, it neither satisfies the inclusion property~\eqref{eq:F-sub-H} of Thm.~\ref{thm:KM-to-abs-CM}, nor any of the commutation relations we have discussed throughout this article. With respect to the former, it is easy to see that $\mathbbm{H}_2 = \text{span}\{\ketbra{+i}{+i}, \ketbra{-i}{-i}\}$ and  $\mathbbm{F}_2 = \text{span}\{\ketbra{+}{+}, \ketbra{-}{-}\}$ holds, where $\ket{\pm i} := 1/\sqrt{2}(\ket{0}\pm \iu \ket{1})$. Since these are the spaces spanned by the eigenvectors of $\sigma_y$ and $\sigma_x$, respectively, neither of them is included within the other. 

With respect to commutation relations, as mentioned above, we have $\widetilde \rho_2(m_1) \propto \Lambda_{2:1}[\ketbra{m_1}{m_1}]$ and $Q_2(m_3) = \Lambda^\dagger_{3:2}[\ketbra{m_3}{m_3}]$, which implies (up to normalisation) 
\begin{gather}
\begin{split}
    \widetilde{\rho}_2(0/1) &= \frac{1}{2}(\ket{0} \pm \iu \ket{1})(\bra{0} \mp \iu \bra{1}) =: \ketbra{i\!\pm}{i\!\pm}\, , \\
    Q_2(0/1) &= \frac{1}{2} (\ket{0}\pm \ket{1})(\ket{0}\pm \ket{1}) =: \ketbra{\pm}{\pm} \, .
\end{split}
\end{gather}
Together with $K_2^{(m_2)} = \ketbra{m_2}{m_2}$, we then obtain, e.g., 
\begin{gather}
\label{eqn::nonvancomm}
    [K_2^{(0)}, Q_2(0)] = \frac{1}{2}(\ketbra{0}{+} - \ketbra{+}{0}) \neq 0\, ,
\end{gather}
i.e., commutativity \`a la L\"uders (and Thm.~\ref{thm:KM-to-abs-CM}) does not hold. Furthermore, as a consequence of Eq.~\eqref{eqn::nonvancomm}, we have $|[K_2^{(m_2)},Q_2(m_3)]| \propto \ident$, such that 
\begin{gather}
    \mathrm{tr}\big[ \widetilde{\rho}_2(m_1)|[K_2^{(m_2)},Q_2(m_3)]|\big] = \mathrm{tr}[\widetilde{\rho}_2(m_1)] \neq 0 \, .
\end{gather}
Thus, for this example, neither commutativity nor absolute commutativity with respect to $\widetilde \rho_2$ hold, and no inclusion property of the relevant spaces is satisfied. \hfill $\blacksquare$\end{example}

Overall, we thus see that, even for the simple case of Markovian dynamics and projective measurements---where necessary and sufficient conditions for classicality are known in terms of NCGD dynamics---no commutation relations between the relevant operators are implied, at least none of the ones discussed in this paper.

\section{Discussion and Conclusion}\label{sec:conclusion}

Throughout this article, we have analysed the connection between classicality and commutativity for Markovian processes probed at multiple points in time. In the two-time setting, it is straight forward to identify the pertinent operators whose commutativity should be assessed. Using the the availability of a full basis of input states, one can then proof an equivalence between commutativity and non-invasiveness, providing a connection between operational and structural notions of classicality. In the multi-time setting, Kolmogorov consistency conditions provide an operationally meaningful notion of classicality; however, it is, a priori, unclear what the relevant operators are to check for commutativity. Here, we have identified the relevant operators, and our work can be seen as a multi-time extension of L{\"u}ders' theorem. As discussed, many crucial assumptions of L{\"u}ders' theorem immediately break down (or become too restrictive) in the multi-time setting, e.g., the guarantee of a full basis of system states at each time. Nonetheless, we have detailed the relevant operators and commutator expressions that imply a connection between operational and structural notions of classicality, putting these distinct notions on a comparable mathematical footing. We have thus overcome a number of complications that arise naturally in physically meaningful scenarios, including probing open system dynamics over multiple times with general quantum measurements. 

In particular, in Sec.~\ref{sec::no_ext_Lud}, we first exemplified how Kolmogorov consistency does not guarantee the vanishing of the analogous commutator expression to the one L{\"u}ders originally considered, and subsequently derived a novel relevant `absolute' commutator expression that indeed implies that satisfaction of Kolmogorov consistency conditions (see Thm.~\ref{thm:abs-CM-to-KM}). Following this, in Sec.~\ref{subsec:classical-com}, we derived additional assumptions such that Kolmogorov consistency implies commutativity (see Thm.~\ref{thm:KM-to-abs-CM}). Lastly, in Sec.~\ref{sec::NCGD}, we connected our results with existing literature to demonstrate the connection between commutativity, classicality, and the ability of the dynamics to generate and detect coherence with respect to sharp measurements in a fixed (but otherwise arbitrary) basis. Along the way, we showed, by way of numerous (counter-)examples that the connection between commutativity and classicality indeed requires more restrictive assumptions in the multi-time case than in the two-time scenario. In turn, this demonstrates that the consideration of multi-time phenomena fundamentally prohibits `simple' extensions of L\"uders' theorem.

Our results provide a connection between commutation relations and the classicality of the observed statistics. However, the absence of necessary \textit{and} sufficient conditions, highlighted via the examples of processes that satisfy none, or just some, of the commutator relations that we identified serve to demonstrate that in the multi-time case, a direct connection between mathematical and operational notions of classicality is far more elusive than in the two-time case (even in the simplest case of projective measurements in a fixed basis). Looking forward, our work opens the door to a number of interesting avenues for exploration. Following our general exposition regarding the structural implications of operational classicality, it would firstly be interesting to identify necessary \textit{and} sufficient conditions for the classicality of observed statistics. While this is a daunting task in general, starting from our considerations, such results might be readily derivable for dynamics that are particularly relevant to certain physical situations: Just as NCGD dynamics equates structural properties to Kolmogorov consistency, we expect it to be possible to derive similarly strong correspondences between particular types of dynamics (e.g., dephasing, depolarising, thermalising, etc.) and the classicality of statistics observed for certain types of instruments (e.g., measure-and-prepare, unital instruments, etc.). Since Kolmogorov consistency ensures the existence of an underlying classical stochastic process that reproduces the statistics correctly, this would in turn shed light on the types of noise that can be effectively replaced by classical environments~\cite{Szankowski2022}, which would have profound impact on the fields of optimal quantum control, reservoir engineering, and the simulation of complex open dynamics. 
\begin{acknowledgments}
We would like to thank {\L}ukasz Cywi{\'n}ski for insightful discussions. F.S. acknowledges support by the Foundation for Polish Science (IRAP project, ICTQT, contract no. 2018/MAB/5, co-financed by EU within Smart Growth Operational Programme). P.T. is supported by the Austrian Science Fund (FWF) project: Y879-N27 (START). S.M. acknowledges funding from the
Austrian Science Fund (FWF): ZK3 (Zukunftkolleg)
and Y879-N27 (START project), the European Union’s Horizon 2020 research and innovation programme under the Marie Sk{\l}odowska Curie grant agreement No 801110, and the Austrian Federal Ministry of Education, Science and Research (BMBWF). The opinions expressed in this publication are those of the authors, the EU Agency is not responsible for any use that may be made of the information it contains.
\end{acknowledgments}

\bibliographystyle{apsrev4-1}

%

\end{document}